\documentclass{article}

\usepackage{amsthm}
\usepackage{amsmath}
\usepackage{amssymb}
\usepackage{enumerate}
\usepackage{xcolor}
\usepackage{listings}
\usepackage{multirow}
\usepackage{color}
\usepackage{hyperref}
\usepackage{doi}
\usepackage{booktabs}
\usepackage{rotating}
\usepackage{bm}
\usepackage{tablefootnote}
\usepackage[normalem]{ulem}
\usepackage{array}
\usepackage{longtable}

\newcommand{\E}{\mathbb{E}}
\newcommand{\ed}{\mathrm{d}}
\newcommand{\R}{\mathbb{R}}
\newcommand{\id}{{\bf 1}}
\renewcommand{\P}{\mathbb{P}}
\newcommand{\Q}{\mathbb{Q}}

\newcommand{\BS}{\mathrm{BS}}

\newcommand{\DT}{\ed {\cal T}(t)}
\newcommand{\DS}{\ed {\cal T}(s)}

\newtheorem{theorem}{Theorem}
\newtheorem{lemma}[theorem]{Lemma}
\newtheorem{remark}[theorem]{Remark}

\newtheorem{proposition}[theorem]{Proposition}
\newtheorem{algorithm}[theorem]{Algorithm}

\newcolumntype{L}[1]{>{\raggedright\let\newline\\\arraybackslash\hspace{0pt}}p{#1}}

\theoremstyle{definition}
\newtheorem{definition}[theorem]{Definition}

\DeclareMathOperator{\tr}{tr}

\DeclareMathOperator{\ri}{ri}

\numberwithin{equation}{section}
\numberwithin{theorem}{section}

\title{Collectivised Pension Investment with Exponential Kihlstrom--Mirman Preferences}
\author{John Armstrong, Cristin Buescu}

\begin{document}

\maketitle

\begin{abstract}
	In a collectivised pension fund, investors agree that any money remaining in the fund
	when they die can be shared among the survivors.
	
	We give a numerical algorithm to compute the optimal investment-consumption
	strategy for an infinite collective of identical investors with exponential Kihlstrom--Mirman preferences,
	investing in the Black--Scholes market in continuous time but consuming in discrete time.
	Our algorithm can also be applied to an individual investor.
	
	We derive an analytic formula for the optimal consumption in the special case of an individual who chooses not to invest in the financial markets. We prove that our problem formulation for a fund with an infinite
	number of members is a good approximation to a fund with a large, but finite number of members.
\end{abstract}

\section*{Introduction}

In a collectivised pension fund, a group of $n$ individuals agree to invest together and to share
any wealth remaining among the survivors when one of the individuals dies. We showed
how to model such funds in \cite{ab-main}. In particular we showed how one can often effectively
reduce the problem of managing a heterogeneous fund of diverse individuals to a series of 
optimization problems for homogeneous funds of identical individuals, thus in this
paper we may safely focus on the management of a homogeneous fund.
The main result of this paper is a numerical method to find the optimal investment
and consumption strategy under certain conditions, namely:
\begin{enumerate}[(i)]
	\item The number of individuals is either $n=1$ or $n=\infty$. These are the most
	important cases for applications, since we will also show that the case for large $n$ is well approximated by the case $n=\infty$.
	\item The fund is investing in a Black--Scholes--Merton market. This is the simplest interesting case.
	\item The individuals preferences are modelled using either
	exponential Kihlstrom--Mirman
	utility functions (see Section \ref{sec:model} for a definition),
	or a deterministically time-varying generalization of exponential Kihlstrom--Mirman preferences.
\end{enumerate} 

As is discussed in detail in \cite{ab-main}, we believe such preferences provide a particularly
good  model for pension investment. For example, one attractive property is that 
exponential Kihlstrom--Mirman preferences remain essentially constant over time: in the terminology of \cite{ab-main} they are {\em stationary} (other
authors sometimes call this property time-consistency). As well as such theoretical properties,
exponential preferences are flexible enough to allow one to model the concepts of
``satiation'', ``risk-aversion'' and ``pension adequacy''.

A standard approach to solving optimal control problems is to work in continuous
time. One then writes down the Hamilton--Jacobi--Bellman (HJB) equation for the problem.
If one is lucky, one can solve this analytically and use this to find the optimal control.
More typically, this is not possible and one must solve the problem numerically. However,
the HJB equation will still provide a guide to developing the appropriate numerical method.

In Section \ref{sec:hjb} we derive the HJB equation for our problem. This gives
a PDE with two space variables and one time variable. However, because exponential preferences
are stationary, they possess a symmetry which allows us to reduce the HJB equation
to a PDE with only one space variable. We are able to solve
this analytically only in the special case where the market provides no return and where $n=1$.
This corresponds to a single investor who hides their life-savings under their mattress, and
gradually spends their savings until they die. We derive this analytic formula in Section \ref{sec:mattress}. Even in this special case we have to resort to using some special functions
to write our formulae.

We see that in general we must use a numerical method. One could apply existing numerical methods
to the HJB equation, for example the methods of \cite{kushnerDupuis}. However, as this has two space dimensions the resulting method would be rather slow.

To avoid this problem, one could attempt to find an existing numerical method for the dimension-reduced equation. The difficulty with this approach is that while the resulting equation looks very similar to an HJB equation, it does not arise directly as the HJB equation of a stochastic control problem. The numerical methods we could find are only designed for equations arising from such control problems. For example, if one were to naively attempt to use the methods of \cite{forsythLabahn} to solve the reduced equation
(ignoring the fact that the requirements they impose on the coefficients do not hold
for our problem), the resulting scheme would be unstable.
Nevertheless, we believe that it should be feasible to apply the ideas behind standard numerical methods in stochastic control to our problem in a way that takes account of the symmetry, and so leads to an efficient numerical scheme.

However, the form of our problem allows for a simpler,
direct approach, so long as we are willing to consider consumption in discrete time and investment
in continuous time. We will see that in this case one can reduce the one-period problem to a line-search problem. This means
that if we know the value function of our optimization problem at time $t$ we can easily approximate
the value function at time $t-\delta t$. We may then proceed inductively to compute the optimal
fund management scheme.  We develop this scheme and prove its convergence in Section \ref{sec:numericalMethod}. We will see that our scheme in fact computes a lower bound for the value function of our problem. This gives a small practical advantage over standard schemes for control problems which can only guarantee that they approximate the value function. However, from a theoretical point of view, this is a significant advantage as it guarantees the stability of our scheme.

In Section \ref{sec:numericalTests} we provide some numerical calculations that validate our results and
our software implementation by comparing with the analytic formula of Section \ref{sec:mattress}.

In Section \ref{sec:exponentialConvergence} we prove that the value function for the
optimal investment problem with $n$ individuals converges to the value function for the
case of $n=\infty$ individuals as $n \to \infty$. This provides a rigorous justification for 
considering the model with $n=\infty$ as an approximate model for a large fund. Moreover,
this is a key ingredient required to show that the investment strategy for heterogeneous
funds described in \cite{ab-main} will be an effective strategy for large fund sizes.

\section{The optimization model}
\label{sec:model}

In this section we review the optimization model for a homogeneous fund given in \cite{ab-main}
in the case of Kihlstrom--Mirman preferences.

Let the set ${\cal T}$ be either
\begin{enumerate}[(i)]
\item a discrete time grid $\{0, \delta t, 2 \delta t, \ldots T-\delta t\}$
      for a fixed grid size $\delta t$ and a terminal time $T$ at which all investors
      are assumed dead;
\item an interval $[0,T)$ with $T \in (0,\infty]$.
\end{enumerate}      
For the case of discrete ${\cal T}$, let $\DT$ be the measure on ${\cal T}$ given by
a Dirac mass at each point of ${\cal T}$. Otherwise let ${\cal T}$ be the Lebesgue measure.

Let the pension outcome $(\gamma, \tau)$ be a pair consisting of a non-negative stochastic cashflow process $(\gamma_t)_{t \in {\cal T}}$ and a random variable $\tau$ taking values in ${\cal T}$ representing an individual's time of death.
We wish to define a preference relation $\preceq$ on such pairs $(\gamma, \tau)$ that describes
the individuals preferences over possible cashflow--mortality outcomes. 
Our convention is that any consumption up to {\em and including} time $\tau$ may effect
the individual's preferences, but any consumption occurring after time $\tau$ will be ignored.

For the purposes
of this paper we will assume that the individual's preferences
are determined by a {\em gain function} of the form
\begin{equation}
{\cal J}(\gamma, \tau) =  \E\left( w \left( \int_0^\tau e^{-bt} u_t( \gamma_t )\, \DT \right) \right).
\label{eq:kihlstromMirmanGain}	
\end{equation}
where $u_t:\R_{\geq 0} \to \R$ is a deterministic time-dependent choice of
concave, increasing utility function,
$w:\R \to \R$ is concave and increasing and $b\in \R$ is a choice of discount rate.
We may then say that the individual prefers outcomes, $(\gamma,\tau)$, that yield higher
values for the gain function, so $(\gamma, \tau) \preceq (\gamma^\prime, \tau^\prime)$ if and only if ${\cal J}(\gamma, \tau) \leq {\cal J}(\gamma^\prime, \tau^\prime)$.

In the case that $u_t$ is time-independent these preferences are called Kihlstrom--Mirman
preferences with mortality after \cite{kihlstromMirman}. If in addition $w(x)=x$ these preferences are called von Neumann--Morgernstern preferences after \cite{vonNeumannMorgernstern}.
If $u_t$ is time-independent $w(x)=-\exp(-x)$ and additionally $b=0$ these preferences
are called exponential Kihlstrom--Mirman preferences with mortality.

Exponential Kihlstrom--Mirman preferences with mortality have a number of attractive properties
which are described in detail in \cite{ab-main}. For example, in the terminology of that paper,
the preferences are {\em stationary} as they do not vary with time. This property is often called time-consistency in the literature. Stationarity is a mathematically attractive property, but
there may be good economic reasons to consider non-stationary preferences. For example, if $\gamma_t$
represents the income from a private pension, and individuals also receive a deterministic state pension which increases in real terms over time, then one might use non-stationary preferences to reflect the reduced need for private pension income as one ages. For this reason it can be useful
to consider the case of time-dependent utility functions $u_t$, in which case one may assume,
without loss of generality, that $b=0$.

\medskip

We wish to find the optimal investment strategy for a collective of $n$ identical individuals
with preferences given by a gain function of the form \eqref{eq:kihlstromMirmanGain}. We assume that the fund
is able to invest in a Black--Scholes--Merton market. This market consists of one
riskless bond with interest rate $r$ and one risky asset $S_t$ which follows the SDE
\[
\ed S_t = S_t(\mu \, \ed t + \sigma \, \ed W_t ), \qquad S_0.
\]
Here $\mu$ and $\sigma$ are constants and $W_t$ is a Brownian motion. We write $(\Omega^M, {\cal F}^M,{\cal F}^M_t, \P^M)$ for the filtered probability space generated by $W_t$.

\medskip

We assume that we have
independent identically distributed random variables $\tau_i$ each representing the time of death of individual $i$. These are assumed to have distribution given by
\begin{equation}
p_t \, \DT.
\label{eq:defpt}
\end{equation}
We will 
write $(\Omega^{L},{\cal F}^{L}, {\cal F}^{L}_t, \P^L)$ for the filtered probability
space generated by the $\tau_i$, the filtration is obtained by requiring that each $\tau_i$ is a stopping time. We will write $F_\tau(t)$ for the distribution function of $\tau$.

We will write $n_t$ for the number of survivors at time $t \in {\cal T}$, that is the number of individuals whose time of death is  greater than or equal to $t$.
This convention ensures that $n_0=n$ and works well with our convention that cashflows
received at the time of death are still consumed.
Note, however, that $n_{t+\delta t}$ will be ${\cal F}^L_t$ measurable.

\medskip

Having selected the gain function for each individual, we must select a gain function
for the fund. We require that the fund is managed such that all surviving individuals receive
the cashflow $\gamma_t$ at times $t \in {\cal T}$. 

In the case that $n$ is finite,
we define a discrete uniformly distributed random variable $\iota$ 
which takes values in $\{1, \ldots, n\}$.
We write $(\Omega^\iota, \sigma^\iota, \P^\iota)$ for the probability space generated by $\iota$. We define a filtration ${\cal F}^\iota_{t \in {\R^+ \cup \{ \infty \}}}$ by
\[
{\cal F}_t^\iota = \begin{cases}
\{ \Omega^\iota, \emptyset \} & t < \infty \\
\sigma^\iota & t= \infty.
\end{cases}
\]
This ensures that $\Omega^\iota$ represents a random choice of individual made at time $\infty$. We take $\Omega = \Omega^M \times \Omega^L \times \Omega^\iota$ 
equipped with the product filtration ${\cal F}_t$ and product measure $\P$.
We then define the gain function for our fund to be
\begin{equation}
{\cal J}^D(\gamma):={\cal J}_\iota(\gamma^\iota, \tau_\iota).
\label{eq:distributionObjective}
\end{equation}

In the case that $n=\infty$ we take
\begin{equation}
{\cal J}(\gamma):={\cal J}_1(\gamma, \tau_1).
\label{eq:robustObjectiveInfinite}
\end{equation}

\medskip

Let $Y_t$ represent the fund
value per individual at time $t$ before consumption or mortality,
and let $\overline{Y}_t$ represent the fund value per individual after consumption.

Our control problem will be to choose a consumption rate $\gamma_t$ and
a proportion $\alpha_t$ to invest in stocks at each time to maximize the gain
of the fund. Let us write down the dynamics of the fund value.

In the continuous time case we have
\begin{equation}
\ed Y_t = Y_t( ((1-\alpha_t)r + \alpha \mu ) \ed t + \alpha_t  \sigma \, \ed W_t )
- \pi_t \gamma_t \, \ed t, \qquad Y_0 = X_0
\label{eq:fundValuePerIndividualCts}
\end{equation}
where $\pi_t$ is the proportion of individuals surviving to time $t$
and $X_0$ is the budget for each individual.
Note that in the case $n=\infty$, $\pi_t=1-F_\tau(t)$ is deterministic.

In the 
discrete time case we have
\begin{equation}
\begin{split}
Y_t &= \begin{cases}
X_0 & \text{t = 0} \\
\lim_{h\to 0+}\overline{Y}_{t-h} &  t \in {\cal T}\setminus\{0\} \\
\overline{Y}_t & \text{otherwise.}
\end{cases} \\
\overline{Y}_t &=
\begin{cases}
Y_t - \pi_t \gamma_t  & t \in {\cal T} \\
\overline{Y}_{t^\prime} +
\int_{t^\prime}^t \overline{Y}_s (((1-\alpha_s)r + \alpha_s \mu) \ed s 
+ \alpha_s \sigma \, \ed W_s )
& t^\prime \in {\cal T} \text{ and } t^\prime\leq t < t^\prime+\delta t.
\end{cases}
\end{split}
\label{eq:fundValuePerIndividual}
\end{equation}

We define an admissible control
to be a progressively measurable process $(\gamma,\alpha)$ such that $Y_t \geq 0$
and $\overline{Y}_t \geq 0$  for all time. We write ${\cal A}$ for the set of admissible controls.

Our objective is to compute
\begin{equation}
v_n = 
\sup_{(\gamma,\alpha) \in {\cal A}} {\cal J}(\gamma),
\label{eq:fundObjective}
\end{equation}
and to find $(\gamma, \alpha)$ achieving (or if necessary, approximating) this supremum.

We note that the problem for the case $n=\infty$ is designed to model
a large finite fund, but to justify this we will need to prove convergence of $v_n\to v_\infty$
as $n\to \infty$. This is done in Section \ref{sec:exponentialConvergence} under mild
assumptions.

\section{The HJB equation}
\label{sec:hjb}

In this section we consider the fully continuous time problem with ${\cal T}=[0,T)$
and $T$ finite.
We will restrict attention to the cases $n=1$ and $n=\infty$. To keep
track of this choice we define a constant $C$
\[
C = \begin{cases}
0 & n=1 \\
1 & n=\infty.
\end{cases}
\]
We will also assume that the function $u_t$ in \eqref{eq:kihlstromMirmanGain} does not
depend upon $t$ and so will write $u_t=u$.

\bigskip

We wish to show that our problem can be reduced to a 
controlled diffusion problem of the type considered in \cite{pham}. Let
us briefly review the elements of this theory that we need.

We
assume we are given a general controlled vector SDE 
\begin{equation}
\ed \bm{X}_s = \bm{b}( \bm{X}_s, \bm{\alpha}_s ) \, \ed s
+ \bm{\sigma}( \bm{X}_s, \bm{\alpha}_s ) \, \ed \bm{W}_s 
\label{eq:generalControlledSDE}.
\end{equation}
We write $\bm{X}_t^{t_0,\bm{X}_{t_0}}$ for the solution of this SDE
with the initial condition $\bm{X}_{t_0}$ at time $t_0$.
We wish to choose a progressively measurable control
$\bm{\alpha}_s$ to maximize our gain function
\begin{equation}
J(t_0,\bm{X}_{t_0}, \bm{\alpha}) := \E( \int_{t_0}^T f(s,\bm{X}_s^{t_0,\bm{X}_{t_0}}) \, \ed s)
\label{eq:gainHJB}
\end{equation}
for some measurable function $f$. Let ${\cal A}$ denote the
set of progressively measurable controls for which this gain function is finite. We define the value function by
\[
v(t_0,\bm{X}_{t_0}) = \sup_{\bm{\alpha} \in {\cal A}} J(t_0, \bm{X}_{t_0}, \bm{\alpha}).
\]
The problem one seeks to solve is to find the value function, and if it exists, the
optimal control $\bm{\hat{\alpha}}$ which achieves the supremum.  We define the diffusion operator ${\cal L}^{\bm a}$ for the diffusion process \eqref{eq:generalControlledSDE} for the control constant $\bm{a}$ by
\[
{\cal L}^{\bm{a}} {v} = \bm{b}(\bm{X},\bm{a}) \cdot
\nabla v + \frac{1}{2} \tr ( \sigma(\bm{X},\bm{a})
\sigma(\bm{X},\bm{a})^\top \nabla^2 v ).
\]
The HJB equation is then defined to be
\begin{equation}
\frac{\partial v}{\partial t}  + \sup_{a} \{ {\cal L}^a v(t, \bm{x}) + f(t,\bm{x}) \} = 0.
\label{eq:hjbGeneral}
\end{equation}
If we can find a solution to this equation with terminal condition
\[
v(T,\bm{x})=0
\]
then the {\em verification theorem} (Theorem 3.5.2 of \cite{pham}) can be used to  demonstrate that, under suitable conditions, $v$ is the value function of the control problem given by the SDE \eqref{eq:generalControlledSDE} with gain function given by \eqref{eq:gainHJB}.

\bigskip

Returning to our specific problem of investment with Kihlstrom--Mirman
preferences, our gain function can be written as
\begin{align}
{\cal J}( \gamma ) &=
\E_\P \left( w \left( \int_0^{\tau_\iota} e^{-bs} u( \gamma_s) \, \ed s \right) \right)  \nonumber \\
&= \E_{\P^M} \left( \int_0^T p_\tau w\left( \int_0^\tau 
e^{-bs} u(\gamma_s ) \ed s  \right)\, \ed \tau \right).
\label{eq:desiredGain}
\end{align}
where $p_t$ was defined in \eqref{eq:defpt}.
Given a progressively measurable control $(\gamma,a)$,
we define processes $U_t$ and $Y_t$ by
\begin{equation}
\begin{split}
U_t &= \int_{t_0}^{t} e^{-bs} u(\gamma_s) \, \ed s + U_{t_0} \\
\ed Y_t &= ((a_t \mu + (1-a_t)r) Y_t - \gamma_t (1-F_\tau(t)^C) \, \ed t + \sigma a_t Y_t \,  \ed W_t, \quad Y_{t_0}.
\label{eq:makeMarkovian}
\end{split}
\end{equation}
The process $U_t$ maintains the state necessary to turn our investment problem into a Markovian problem. The process $Y_t$ corresponds to the fund value per individual defined in the previous section.

We define a vector process $\bm{X}_t$ by
\[
\bm{X}_t = (U_t,Y_t)
\]
Let us define the gain function for controlled diffusion problem with 
initial condition $(U_{t_0}, Y_{t_0})$ to be
\begin{equation}
J_{t_0,(U_{t_0}, Y_{t_0}),(\gamma,a)} = \E_{\P^M}
\left( \int_{t_0}^T f(s,\bm{X}_s) \ed s \right)
\label{eq:gainWithU}
\end{equation}
with
\[
f(t,(U,Y)) = p_t w(U).
\]
This coincides with our the gain function \eqref{eq:desiredGain} when $t_0=0$ and $U_{t_0}=0$. We have now written our problem in the desired form.

We can now write down the HJB equation.
\begin{multline}
\frac{\partial v}{\partial t} + \sup_{a,\gamma} \Big\{
e^{-bt} u(\gamma) \frac{\partial v}{\partial U} + \\
((a \mu + (1-a) r)Y - \gamma (1-F_\tau(t))^C ) \frac{\partial v}{\partial Y}
+ \frac{1}{2} \sigma^2 a^2 Y^2 \frac{\partial^2 v}{\partial Y^2}
+ p_t w(U) \Big\} = 0
\label{eq:hjbFull}
\end{multline}

As we have seen, this has two space variables because we need an extra space variable to address the non-Markovianity of Kihlstrom--Mirman preferences.

In the case of exponential preferences, we expect to be able
to reduce the dimension as the preferences will be Markovian.
So let us now assume that $w(x)=-\exp(-x)$. We now
have the following symmetry of our gain function:
\[
J_{t_0,(U+c, Y)}(\gamma) = e^{-c} J_{t_0,(U, Y)}(\gamma).
\]

Also note that scaling of the value function $v$ at intermediate
times $t_0$ is not particularly natural. If we define
$\tilde{v}$ to be the value of the optimal investment strategy starting at time $t_0$ with investment $Y_{t_0}$ then we will have
\begin{equation}
\tilde{v}(t_0,Y_{t_0}):= \sup_{(\gamma_t,a_t) \in {\cal A}} \E_{\P^\iota} \left(
w \left( \int_{t_0}^{\tau_\iota} e^{-bs} u(\gamma_s) \ed s \right)
\mid \tau_\iota \geq t_0 \right).
\label{eq:defvtilde}
\end{equation}
We see that
\[
\tilde{v}(t,Y) = \frac{1}{1-F_\tau(t_0)} v(t_0,(0,Y)).
\]

Combining our observations on the symmetry of the gain function with the
scaling behaviour of $v$ motivates the definition
\[
v(t,(U,Y)) = (1-F_\tau(t)) e^{-U} \hat{v}(t,Y)
\]
for some function $\hat{v}(t,Y)$. Substituting this into
the full HJB equation \eqref{eq:hjbFull}, we get the dimension
reduced HJB equation
\begin{multline}
\frac{\partial \hat{v}}{\partial t} + \sup_{a,\gamma} \Big\{
-\lambda(t) \hat{v}
- e^{-bt} u(\gamma) \hat{v}  - \lambda(t) \\
+ ((a \mu + (1-a) r)Y - \gamma (1-F_\tau(t))^C ) \frac{\partial \hat{v}}{\partial Y}
+ \frac{1}{2} \sigma^2 a^2 Y^2 \frac{\partial^2 \hat{v}}{\partial Y^2}
\Big\}=0
\label{eq:hjbreduced}
\end{multline}
where $\lambda(t)$ is the force of mortality
\[
\lambda(t) = \frac{p_t}{1-F(t)}.
\]

In summary, if we can find a smooth, non-positive, solution to the reduced HJB equation \eqref{eq:hjbreduced} with terminal condition $\hat{v}=-1$,
it will give rise to a smooth solution of the full HJB equation \eqref{eq:hjbFull}. One can then apply the verification theorem to show that this must be equal to the value function of the full problem.
Hence $\tilde{v}=\hat{v}$ where $\tilde{v}$ is the value function
defined in equation \eqref{eq:defvtilde}.

\medskip

We also expect a dimension reduction on the case of von Neumann-Morgernstern preferences. So we now assume $w(x)=x$.
We now
have the following symmetry of our gain function:
\[
J_{t_0,(U+c, Y)}(\gamma) = 
\int_{t_0}^T c p_s \, \ed s + J_{t_0,(U, Y)}(\gamma)
= (1-F_\tau(t)) c + J_{t_0,(U, Y)}(\gamma).
\]
Hence we may write the value function as
\[
v(t,(U,Y)) = (1-F_\tau(t)) U + \hat{v}(t,Y)
\]
for some function $\hat{v}(t,Y)$.  Substituting this into
the full HJB equation \eqref{eq:hjbFull}, we get the dimension
reduced HJB equation
\begin{multline}
\frac{\partial \hat{v}}{\partial t} + \sup_{a,\gamma} \Big\{
e^{-bt} u(\gamma) (1-F(t))   \\
+ ((a \mu + (1-a) r)Y + \gamma (1-F_\tau(t))^C ) \frac{\partial \hat{v}}{\partial Y}
+ \frac{1}{2} \sigma^2 a^2 Y^2 \frac{\partial^2 \hat{v}}{\partial Y^2}
\Big\}=0.
\label{eq:hjbreducedVNM}
\end{multline}

It is interesting to note that in the individual case, $C=0$,
the discount factor $e^{-bt}$ and the life-expectancy
term $(1-F(t))$ both appear only as coefficients of $u(\gamma)$.
Thus the solution for the individual problem with no discount 
factor and an exponential mortality distribution with force of
mortality $\lambda$ is equivalent to the individual problem with
no mortality, but with a discount rate of $\lambda$. 
This provides some
evidence that it may not be appropriate to use discounting in a
model which already features mortality. This is an important
point in deciding whether to use exponential Kihlstrom--Mirman
preferences or homogeneous Epstein--Zin preferences, as is discussed
in detail in \cite{ab-main}.

\section{A special case}
\label{sec:mattress}

In general for exponential Kihlstrom--Mirman preferences
we expect that we will need to solve the HJB
equation \eqref{eq:hjbreduced} numerically. However, it will be helpful to have
an analytical solution in a special case, if only to verify
the accuracy of our numerical solution.

If we take the case $n=1$ and assume that mortality is given by
an exponential distribution with constant force of mortality $\lambda$,
then our problem will be translation invariant in time. We introduced
the variable $\hat{v}$ earlier to preserve this symmetry, so we may
therefore look for a steady state solution to the HJB equation
satisfying $\frac{\partial \hat{v}}{\partial t}=0$. The HJB equation
will then become an ordinary differential equation (ODE).

The resulting ODE will still not be particularly simple, so we assume further that $\mu=r=b=0$.
This implies that the optimal solution will satisfy $a=0$. This is intuitively
clear, a rigorous argument for the corresponding result for von Neumann Morgernstern utility
is given in \cite{ab-main} and the same argument can be applied to this case. Making the appropriate simplifications to \eqref{eq:hjbreduced} we obtain:
\begin{equation}
\sup_\gamma \left\{ -\lambda \hat{v} -\lambda - u(\gamma) \hat{v} -  \, \gamma \hat{v}^\prime \right\} = 0.
\label{eq:hjb1}
\end{equation}
Technically one should derive this ODE from first principles using the
techniques of Section \ref{sec:hjb}
in the infinite horizon case, but doing so is unrevealing.

Our final specialization is to assume
\begin{equation}
u(x) = a \, x^k + c.
\label{eq:poweru}
\end{equation}
with $0 < k < 1$.

Although these assumptions are very restrictive, this case has a certain charm. It
represents an investor with exponential preferences who eschews the financial system
and keeps their money under the mattress, deciding only upon when to consume.

Let $\gamma^*$ denote the value of $\gamma$ achieving the supremum in equation \eqref{eq:hjb1}. We find:
\[
u^\prime(\gamma^*) = - \frac{\hat{v}^\prime}{\hat{v}}
\]
So
\begin{equation}
a k (\gamma^*)^{k-1}= -\frac{\hat{v}^\prime}{\hat{v}}.
\label{eq:gammaStar1}
\end{equation}
Substituting this into \eqref{eq:hjb1} and using our expression \eqref{eq:poweru} for $u$, we find
\[
a (1-k) (\gamma^*)^k+c+\lambda +\frac{\lambda }{\hat{v}}=0.
\]
We may solve this to find that $\gamma^*$ satisfies
\begin{equation}
\gamma^* = \left(\frac{c \hat{v}+\lambda +\lambda  \hat{v}}{a (k-1) \hat{v}}\right)^{1/k}.
\label{eq:gammaInTermsOfV}
\end{equation}
Substituting into \eqref{eq:gammaStar1} we have
\begin{equation}
ak \left(\frac{c \hat{v}+\lambda +\lambda  \hat{v}}{a (k-1) \hat{v}}\right)^{\frac{k-1}{k}} = -\frac{\hat{v}^\prime}{\hat{v}}.
\label{eq:forv}
\end{equation}
To solve this ODE, we need to consider the initial condition. We know
that if our initial wealth is $0$, consumption will be zero at all times. Our choice of utility \eqref{eq:poweru} then ensures
that
\[
\hat{v}(0)= \E\left( -\exp\left(-\int_0^\tau c \, \ed t \right)\right) =
-\int_0^\infty e^{-cs} \lambda e^{-\lambda s} \, \ed s =  -\frac{\lambda}{c+ \lambda}.
\]

To simplify \eqref{eq:forv} and this initial condition, we define
\begin{equation}
w := c\hat{v} + \lambda + \lambda \hat{v}
\label{eq:defw}
\end{equation}
and rewrite \eqref{eq:forv} as
\[
\left(
\frac{(1-k)^{1-\frac{1}{k}} a^{-1/k} w^{-\frac{k-1}{k}} (c+\lambda )^{\frac{1}{k}-1} (\lambda -w)^{-1/k}}{k}
\right) \frac{\ed w}{\ed x} = 1.
\]
Hence
\begin{align*}
x &= \int_0^w \left(
\frac{(1-k)^{1-\frac{1}{k}} a^{-1/k} W^{-\frac{k-1}{k}} (c+\lambda )^{\frac{1}{k}-1} (\lambda -W)^{-1/k}}{k} \ed W
\right) \\
&= (1-k)^{1-\frac{1}{k}} a^{-1/k} \lambda ^{-1/k} w^{1/k} (c+\lambda )^{\frac{1}{k}-1} \, _2F_1\left(\frac{1}{k},\frac{1}{k};1+\frac{1}{k};\frac{w}{\lambda }\right).
\end{align*}
In this expression $_2F_1$ is the Gauss hypergeometric function
and we have simply used a standard integral identity for this function.

Using our definition for $w$ we find the following implicit equation for $\hat{v}$
\begin{multline}
x = 
(1-k)^{1-\frac{1}{k}} a^{-1/k} \lambda ^{-1/k} (c+\lambda )^{\frac{1}{k}-1} (\hat{v} (c+\lambda )+\lambda )^{1/k} \\
\times \, 
_2F_1\left(\frac{1}{k},\frac{1}{k};1+\frac{1}{k};\frac{\lambda +\hat{v} (c+\lambda )}{\lambda }\right).
\label{eq:analyticexpressionforv}
\end{multline}

We would also like to compute how the consumption $\gamma^*$ varies over time. To do this note that the wealth $x$ satisfies
\[
\frac{\ed x}{\ed t} = -\gamma.
\]
It follows that
\[
\frac{ \ed \hat{v}}{\ed t} = -\gamma \frac{\partial \hat{v}}{\partial x}.
\]
This allows us to rewrite \eqref{eq:gammaStar1} as
\[
a k \gamma^{k-1}=\frac{\frac{\ed \hat{v}}{\ed t}}{\gamma  \hat{v}}.
\]
Hence from \eqref{eq:forv} we find
\[
\frac{\ed \hat{v}}{\ed t} = \frac{k (c \hat{v}+\lambda +\lambda  \hat{v})}{k-1}
\]
In terms of $w$
\[
\frac{\ed w}{\ed t}= (c+\lambda)\frac{k}{k-1} w.
\]
Solving this, we obtain the dynamics for $w$ or equivalently $\hat{v}$. The initial condition at $t=0$ is $w=w_0$,
which gives
\[
\log(w) - \log(w_0)=  (c+\lambda)\frac{k}{k-1} t
\]
Then using
\eqref{eq:gammaInTermsOfV} we obtain the dynamics of the optimal consumption $\gamma^*$. The result is
\begin{equation}
\gamma^*_t = 
(1-k)^{-1/k} a^{-1/k} \lambda ^{-1/k} (c+\lambda )^{1/k} \left(w_0 e^{\frac{k t (c+\lambda )}{k-1}}\right)^{1/k} \left(1-\frac{w_0 e^{\frac{k t (c+\lambda )}{k-1}}}{\lambda }\right)^{-1/k}.
\label{eq:analyticexpressionforgamma}
\end{equation}

We summarize our results.
\begin{proposition}
	For the investment problem for:
	\begin{enumerate}[(i)]
		\item a single individual;
		\item in a market with $\mu=r=0$;
		\item with an exponential mortality with intensity $\lambda$;
		\item with exponential Kihlstrom--Mirman preferences with $u$ given by \eqref{eq:poweru} and $b=0$;
	\end{enumerate}
	the value function $\hat{v}$ at time $0$ satisfies \eqref{eq:analyticexpressionforv} where $x$ is the initial wealth.
	Defining $w$ by \eqref{eq:defw}, the optimal consumption $\gamma^*$
	satisfies \eqref{eq:analyticexpressionforgamma}.	
\end{proposition}	

The one new theoretically interesting observation that has emerged is that the adequacy level will affect the value function and the consumption.

One can also obtain an analytic formula if one takes $k<0$ and $a<0$ in \eqref{eq:poweru}, the difference being that one should now take $x=\infty$ as the boundary condition to ensure that $u(x)=0$.

\section{Numerical solution for discrete time consumption}
\label{sec:numericalMethod}

We describe a numerical method for solving the problem when consumption takes
place in discrete time but investment occurs in continuous time.

\subsection{Solution to the one period problem}

We will be interested in solving the investment problem given by
\eqref{eq:fundObjective} in the cases when $n=1$ and $n=\infty$.
As before, we define $C=0$ when $n=1$, and $C=1$ when $n=\infty$. Our numerical
method extends straightforwardly to the case when the function $u_t$
in \eqref{eq:kihlstromMirmanGain}
depends on both consumption and time, but to ease notation we will
only give the formulae for time-independent $u_t=u$.

Write ${\cal A}_{X,t}$ for the admissible consumption-investment
strategies that start with wealth $X$ at time $t$.
Define the value function $v$, as
a function of initial wealth, $X$ at time  $t_1 \in {\cal T}$ by
\[
v_{t}(X):=\sup_{{\gamma,\alpha} \in {\cal A}_{X,t}}
\E\left(-\exp\left(- \int_{t}^\tau u(\gamma_s) \, \DS \right)
\mid \tau \geq t
\right).
\]

Given $v_{t}(X)$, we wish to compute $v_{t-\delta t}(X)$, we will then be able
to recursively compute $v_{t}$ for all $t \in {\cal T}$. Our next theorem
shows how to compute $v_{t-\delta t}(X)$, but
in order to state our results concisely we first make the following definitions.
\begin{definition}
	Let $u:\R \to \R \cup \{ \pm \infty \}$ be concave and increasing. Define
	\[
	u^\dagger(p):\R_{>0} \to \R
	\]
	by 
	\[
	u^\dagger(p) = \inf \{ x \mid p \in \partial u(x) \}
	\]
	where $\partial u(x)$ is the sub-differential of $u$ at $x$.
	\label{def:dagger}	
\end{definition}

For sufficiently regular
functions $u$, we have $u^\dagger = (u^\prime)^{-1}$, or, equivalently, $u^\dagger$
is the derivative of the Legendre transform of $u$. Since the Legendre transform of $u$
naturally arises if we apply convex duality to our problem, $u^\dagger$ will also arise
naturally from the Kuhn--Tucker conditions for the dual problem.

\begin{definition}
Define
\begin{equation}
Q(u):=\Phi\left( M  + \Phi^{-1}(u) \right),
\label{eq:defQ}
\end{equation}
$\Phi$ is the cumulative distribution function of the
standard normal distribution and
\[
M:= \frac{\left| \mu-r \right|\sqrt{\delta t}}{\sigma }.
\]
Define
\begin{equation}
q^A_{\BS}(u)= \frac{\ed Q}{\ed u}.
\end{equation}
\end{definition}

As we will show in Lemma \ref{lemma:bsmCanonicalForm} below, the quantity $q^A_{\BS}(u)$ can be related to the pricing kernel of the Black--Scholes model.

We may now state the following result which allows us to solve the one period problem.

\begin{proposition}
	Suppose that $t_1=t_0+\delta t$ and that  $v(X):=v_{t_1}(X)$ is known, concave and increasing
	for $X>0$, equal to $-\infty$ for $X\leq 0$,
	and satisfies $v(X) \leq 0$.
	\begin{enumerate}[(a)]
	\item 	\label{prop:solutionExponentialA} $v_{t_0}(X)$ is itself concave and increasing for $X>0$, equal to $-\infty$ for $X \leq 0$
	and satisfies $v(X) \leq 0$.	
	\item \label{prop:solutionExponentialB} For each $\eta>0$ define a function on 
	$f^\eta:(0,1)\to \R_{\geq0}$ by
	\begin{equation}
	f^\eta(s) = v^\dagger \left( 
	\eta  s^{C-1}_{t_0} e^{-r \delta t} q^A_{\BS}(s) \right).
	\label{eq:defFEta}
	\end{equation}
	Define $\gamma^\eta \in \R_{\geq0}$ by
	\begin{equation}
	\gamma^\eta  = u^\dagger \left( -\frac{\eta}{\delta t}
	\left(
	-1
	+ s_{t_0} \int_0^1 (1+v(f^\eta(s))) \, \ed s \right)^{-1}
	\right).
	\label{eq:defGammaEta}
	\end{equation}
	Define $X^\eta$ by
	\begin{equation}
	X^\eta = \gamma^{\eta} + s^C_{t_0} \int_0^1 e^{-r \delta t} q^{A}_{\BS}(s) f^{\eta}(s).
	\label{eq:defXeta}
	\end{equation}
	If there exists $\eta_{X_0}$ such that $X^{\eta_{X_0}}=X_0$ then we have
	\[
	v_{t_0}(X_0) = 
	\exp( -u(\gamma^{\eta_{X_0}})  \delta t ) \left(
	-1
	+ s_{t_0} \int_0^1 (1+v(f^{\eta_{X_0}}(s))) \, \ed s \right)
	\]
	and $\gamma^{\eta_{X_0}}$ is the optimal consumption at time $t_0$.		
	\end{enumerate}
	\label{prop:solutionExponential}
\end{proposition}

Part \eqref{prop:solutionExponentialA} is trivial. For example the
statement about concavity follows from \cite{luenberger} Proposition 8.3.1.
The proof strategy for Part \eqref{prop:solutionExponentialB} is as follows:
\begin{enumerate}[(i)]
	\item Use the dynamic programming principle to obtain a recursive formulation of the problem.  This is done in Lemma \ref{lemma:dpp}
	\item Reduce the continuous time investment problem of the recursion step to a calculus of variations problem using the classification of one-period complete markets. This is done in Lemma \ref{lemma:reduceToCalculusVariations}, and is the main novel step in our approach.
	\item Solve the resulting calculus of variations problem. This is done in Lemma \ref{lemma:solveCalculusOfVariations}.
\end{enumerate}

Let us first see how to compute 
$v_{t_0}(X_0)$ as the solution to a one period optimal investment problem.

\begin{lemma}
	\label{lemma:dpp}
	Assume the conditions of Proposition \ref{prop:solutionExponential}.
	Let ${\cal A}_{X_0,t_0,t_1}$ denote the set of pairs $(\gamma_{t_0},\alpha)$
	where $\alpha$ is an admissible investment strategy for the period $[t_0,t_1]$ and $\gamma_{t_0}\in \R$ is the consumption at time $t_0$
	and satisfies $\gamma_{t_0}<X$. Then
	\begin{equation}
	\begin{split}
	v_{t_0
	}(X_0):=\sup_{{\gamma_{t_0},\alpha} \in {\cal A}_{X_0,t_0,t_1}}
	\Big\{
	\exp\left(-u(\gamma_{t_0}) \delta t \right)
	\left( -1 + s_{t_0} 
	\E\left( 1+v_{t_1}(X^{(\gamma_{t_0},\alpha)}_1) \right)
	\right)
	\Big\}
	\label{eq:exponentialDynamicProgramming}
	\end{split}
	\end{equation}
	where $X^{(\gamma_{t_0},\alpha)}_1$ is the value obtained by 
	following the investment strategy $\alpha$ from $t_0$ to $t_1$
	with an initial wealth of $s_t^{-C}(X_0-\gamma_{t_0})$.
\end{lemma}
\begin{proof}
	We calculate
	\begin{align*}
	v_{t_0}(X_0)
	&=\sup_{{\gamma,\alpha} \in {\cal A}_{X_0,t_0}} \Big\{ 
	\E\left(-\exp\left(-u(\gamma_{t_0}) \delta t \right) \P(\tau<t_1 \mid \tau \geq t_0 )
	\right)
	\\
	&\quad \quad \quad 
	+ \E\left(-\exp\left(-u(\gamma_{t_0}) \delta t - \int_{t_1}^\tau u(\gamma_t) \, \DT \right  ) \mid \tau\geq t_1 \right) \P( \tau \geq t_1 \mid \tau \geq t_0 )
	\Big\} \\
	&=\sup_{{\gamma,\alpha} \in {\cal A}_{X_0,t_0}} \Big\{
	-  (1-s_{t_0})\exp\left(-u(\gamma_{t_0}) \delta t \right) \\
	&\quad \quad \quad +
	s_{t_0} \exp(-u(\gamma_{t_0}) \delta t) \E\left(-\exp\left( - \int_{t_1}^\tau u(\gamma_t) \, \DT \right  ) \mid \tau \geq t_1 \right) \Big\} \\
	&=\sup_{{\gamma,\alpha} \in {\cal A}_{X_0,t_0}} 
	\Bigg\{
	\exp\left(-u(\gamma_{t_0}) \delta t \right) \times \\
	&\quad \quad \quad
	\left(
	-  1 +
	s_{t_0} \E\left(1-\exp\left( - \int_{t_1}^\tau u(\gamma_t) \, \DT \right  ) \mid \tau \geq t_1 \right) \right) \Bigg\}
	\end{align*}
	The result now follows by the dynamic
	programming principle.
\end{proof}

Equation \eqref{eq:exponentialDynamicProgramming} is a one-period investment problem in a complete market.
Complete one period markets are classified 
in \cite{armstrongClassification}. This allows us to find
a more convenient, but isomorphic, representation
of our market. 
For complete one period markets, we may
say that two markets are isomorphic if they have the same
risk-free rate and if there is a map which acts as a 
probability space isomorphism for both the $\P$ and $\Q$
measures simultaneously.

Let $\Omega^A$ be the probability space given by $[0,1]\times[0,1]$
equipped with the Lebesgue measure. Let $q^A:[0,1]\to \R_{>0}$ be a
measurable function of integral $1$. We may define an
abstract financial market
$(\Omega^A,q^A,r)$ whose assets consist of random variables $f$
(representing the payoff of the asset)
defined on $\Omega^A$. The cost of asset $f$ is given by
\[
C^A(f):= \int_{[0,1]\times[0,1]} e^{-r \delta t} f(x,y)\,q^A(x) \, \ed x \, \ed y
\]
if this integral exists. Assets of positively infinite or undefined cost
cannot be purchased. Assets of infinitely negative cost can
be purchased at any price. The $A$ in our superscripts
stands for abstract. Notice that in this abstract market
the random variable $U$ defined by $U(x,y)=x$ is uniform
in the $\P^A$ measure and has density $q^A$ in the $\Q^A$ measure.

\begin{lemma}
	\label{lemma:bsmCanonicalForm}	
	As a one period market, the Black--Scholes--Merton market 
	from time $t_0$ to time $t_1$
	is isomorphic to the market $(\Omega^A,q^A_{\BS},r)$.
\end{lemma}
\begin{proof}
	
	If $\mu=r$, then the result is trivial. We will consider
	the case $\mu>r$, the case $\mu<r$ is similar.
	
	The classification of complete markets already shows that
	the Black--Scholes--Merton market over the time period $[t_0,t_1]$
	is isomorphic to a
	market of this form for an appropriate choice of $q^A$
	which we will call $q^A_{\BS}$. Let
	$\frac{\ed \Q}{\ed \P}$ denote the Radon--Nikodym derivative
	of the measures $\Q$ and $\P$ in the Black--Scholes--Merton market.
	Let $F_\frac{\ed \Q}{\ed \P}$ denote the $\P$-measure distribution
	function of the Radon--Nikodym derivative.
	The classification theorem moreover gives us
	an isomorphism
	for both the $\P$ and $\Q$ measures which maps the uniformly distributed random variable $U^\prime:=F_\frac{\ed \Q}{\ed \P}(\frac{\ed Q}{\ed P})$
	to $U$. In particular this tells us that 
	\begin{equation}
	\int_{0}^{w} q^A_{\BS}(s) \ed s = \P_{\Q^A}(U \leq w)
	= \P_{\Q}(F_\frac{\ed \Q}{\ed \P}(U^\prime \leq w) )
	\label{eq:defQInProof}
	\end{equation}
	Differentiating this, we may obtain an expression for $q^A_{\BS}$.		
	
	The $\P$ measure distribution function of the log stock price, $z_{t_1}=\log(S_{t_1})$ given the log stock price $z_{t_1}$
	in the Black--Scholes--Merton model is
	\[
	p(z) = \frac{1}{\sqrt{2 \pi \sigma \delta t }}
	\exp\left( -\frac{(z-(z_{t_0}+(\mu-\frac{1}{2}\sigma^2)\delta t))^2}{2 \sigma^2 \delta t}
	\right).
	\]
	Similarly the
	$\Q$ measure distribution function of $z_{t_1}$ is
	\[
	q(z) = \frac{1}{\sqrt{2 \pi \sigma \delta t }}
	\exp\left( -\frac{(z-(z_{t_0}+(r-\frac{1}{2}\sigma^2)\delta t))^2}{2 \sigma^2 \delta t}
	\right).
	\]
	The standard computation of the $\Q$ measure using Girsanov's theorem
	shows that
	\begin{align*}
	\frac{\ed \Q}{\ed \P}(z)&=\frac{q(z)}{p(z)}.
	\end{align*}
	Hence
	\begin{align*}
	\frac{\ed \Q}{\ed \P}(z) &= \exp\left( -\frac{(z-(z_{t_0}+(r-\frac{1}{2}\sigma^2)\delta t))^2
		- (z-(z_{t_0}+(\mu-\frac{1}{2}\sigma^2)\delta t))^2}{2 \sigma^2 \delta t}
	\right).
	\end{align*}
	Note that the term in side the $\exp$ is linear in $z$, so $\frac{\ed \Q}{\ed \P}$
	is decreasing. Hence
	$U^\prime(z)$ is decreasing, and we recall that $U^\prime$ is uniformly distributed. Hence, $U^\prime(z)=1-F_z(z)$ where $F_z$ is the $\P$-measure
	distribution function of $z_t$. But conditioned on $z_{t_0}$, $z_{t_1}$ is normally distributed with mean $\mu-\tfrac{1}{2}\sigma^2$ and
	standard deviation $\sigma \sqrt{\delta t}$. Hence
	\[
	z_{t_1} = z_{t_0} + (\mu-\tfrac{1}{2}\sigma^2) \delta t + 
	\sigma \sqrt{\delta t} \, \Phi^{-1}(U^\prime)
	\]
	where $\Phi$ is the inverse distribution function of the standard normal
	distribution.
	
	We now compute
	\begin{align*}
	\P_\Q( U^\prime \leq w )
	&=\P_\Q( z_{t_1} \leq
	z_{t_0}+(\mu-\tfrac{1}{2}\sigma^2)\delta t + 
	\sigma \sqrt{\delta t} \Phi^{-1}( w ) ) \\
	&=\P_\Q( z_{t_1} \leq
	z_{t_0}+(r-\tfrac{1}{2}\sigma^2)\delta t
	+ (\mu-r) \delta t
	+
	\sigma \sqrt{\delta t} \Phi^{-1}( w ) ).
	\end{align*}
	Since $z_{t_1}$ is normally distributed in the $\Q$ measure with 
	mean $r-\tfrac{1}{2}\sigma^2$ and standard deviation $\sigma \sqrt{\delta t}$
	we find
	\[
	\P_\Q( U^\prime \leq w ) = \Phi\left( \left| \frac{(\mu-r)\sqrt{\delta t}}{\sigma}
	\right|  + \Phi^{-1}(w) \right).
	\]
	Combining this with \eqref{eq:defQInProof}, we get the result.
\end{proof}

Having found a simple isomorphic representative of our market, we can rewrite
the equation \eqref{eq:exponentialDynamicProgramming} in terms of the abstract market $\Omega^A$.

\begin{lemma}
	\label{lemma:reduceToCalculusVariations}
	Assume the conditions of Proposition \ref{prop:solutionExponential}.
	The value function $v_{t_0}(X_0)$ can be calculated by
	solving the optimisation problem
	\begin{equation}
	\begin{aligned}
	& \underset{\gamma \in \R, f \in L^0[0,1]}{\mathrm{maximize}}
	& & 
	\exp( -u(\gamma)  \delta t ) \left(
	-1
	+ s_{t_0} \int_0^1 (1+v(f(s))) \, \ed s \right) \\
	& \text{subject to}
	& & \gamma + s^C_{t_0} \int_0^1 e^{-r \delta t} q^{A}_{\BS}(s) f(s) \, \ed s \leq X_0.
	\end{aligned}
	\label{eq:abstractInvestmentProblem}
	\end{equation}
	taking $v=v_{t_1}$.	
\end{lemma}
\begin{proof}
	Let us write $(\gamma_{t_0}, f)$ for a pair of a consumption
	$\gamma_{t_0} \in \R$ and an
	investment $f \in L^0(\Omega^U)$.
	We denote by ${\cal B}_{X_0}$ the set 
	of consumptions and investments that are available with
	a budget of ${X_0}$
	\[
	{\cal B}_{X_0}=\{ (\gamma_{t_0},f) \in \R \times L^0(\Omega^A) \mid \gamma_{t_0} + s_{t_0}^C C^A_{\BS}(f) \leq X_0 \}.
	\]
	If we also write
	\begin{equation*}
	\begin{split}
	{\cal J}^A_{t_0}(\gamma_{t_0},f):=
	&\exp( -u(\gamma_{t_0})  \delta t )
	\left(-1 + s_{t_0}  \int_{[0,1]\times[0,1]} 
	\, (1+v_{t_1}( f(x,y) )) \, \ed x \, \ed y
	\right)
	\end{split}
	\end{equation*}
	to accord with equation \eqref{eq:exponentialDynamicProgramming},
	then the fact that our markets are isomorphic allows
	us to deduce that
	\begin{equation}
	v_{t_0}(X_0):=\sup_{(\gamma_{t_0},f) \in {\cal B}_{X_0}} {\cal J}^A(\gamma_{t_0},f)
	\label{eq:exponentialDynamicProgramming2}.
	\end{equation}
	
	Since $v_{t_1}$ is assumed to be concave we may average an
	investment $f(x,y)$ over the factor $y$ to obtain a new investment
	$\overline{f}$ which achieves a higher value for the gain
	function ${\cal J}^A$. Thus we may restrict our attention
	to investments $f(x,y)$ which depend only upon $x$. The result follows.
\end{proof}

Note that an investment $f \in L^1$ for this abstract market
model corresponds to a derivative with payoff given by
the random variable $f(F_{\frac{\ed \Q}{\ed \P}} (\frac{\ed \Q}{\ed \P}) )$  in the original Black--Scholes--Merton market (or indeed in any isomorphic market). This derivative can then be replicated by delta hedging
in the Black--Scholes--Merton market. So the solution to the
abstract investment problem \eqref{eq:abstractInvestmentProblem}
can be straightforwardly mapped to a solution of the original problem.

\begin{lemma}
	\label{lemma:solveCalculusOfVariations}
	Assume the conditions and definitions of Proposition \ref{prop:solutionExponential}.	
	If an $\eta_{X_0}$ exists with $X^{\eta_{X_0}}=X_0$,
	then the solution of \eqref{eq:abstractInvestmentProblem} is given by $f^{\eta_{X_0}}$
	and $\gamma^{\eta_{X_0}}$.
\end{lemma}
\begin{proof}
	We will now solve \eqref{eq:abstractInvestmentProblem} using
	the method of Lagrange multipliers. We define a vector space $V=\R \oplus L^0([0,1]) \oplus \R$
	For $\lambda \in \R$,
	we define the Lagrangian
	$L:V \to \R$ by
	\begin{equation}
	\begin{split}
	L(\gamma,f,\lambda):=&
	\exp( -u(\gamma)  \delta t ) \left(
	-1
	+ s_{t_0} \int_0^1 (1+v(f(s))) \, \ed s \right) \\
	& + \lambda \left( -X_0 + \gamma +  s^C_{t_0} \int_0^1 e^{-r \delta t} q^A_{\BS}(s) f(s) \, \ed s \right).
	\label{eq:defLagrangian}
	\end{split}
	\end{equation}
	
	Computing the directional derivatives of $L(\gamma,f)$ in we find the following
	necessary and sufficient conditions for $(\gamma,f)$ to be a saddle point of $L(\gamma,f,\lambda)$
	for the given $\lambda$. Firstly
	\begin{equation}
	0 \in 
	-\partial u(\gamma) \delta t \exp( -u(\gamma)  \delta t ) \left(
	-1
	+ s_{t_0} \int_0^1 (1+v(f(s))) \, \ed s \right) + \lambda
	\label{eq:kuhnTucker1}
	\end{equation}
	where $\partial u(\gamma)$ is the subdifferential of $u$ at $\gamma$.
	Secondly
	\[
	0 = \int_0^1 \left( \exp( -u(\gamma)  \delta t ) s_{t_0} (\partial v)(f(s)) + \lambda  s^C_{t_0} e^{-r \delta t} q^A_{\BS}(s) \right) g(s) \, \ed s.
	\]
	The integral is well-defined since $\partial v$ will
	be single valued almost everywhere.
	This must hold for all $g(s)$ so this is equivalent to requiring
	\begin{equation}
	(\partial v)(f(U)) =- \lambda \exp( u(\gamma)  \delta t )  s^{C-1}_{t_0} e^{-r \delta t} q^A_{\BS}(U).
	\label{eq:kuhnTucker2}
	\end{equation}
	for almost all $U\in(0,1)$.
		
	If the Kuhn-Tucker
	conditions \eqref{eq:kuhnTucker1} and \eqref{eq:kuhnTucker2}
	are satisfied, $((\gamma,f),\lambda)$ will be a saddle point of the Lagrangian.
	The theory of Lagrange multipliers (see \cite{rockafellar} Theorem 28.3)
	now shows that if we can
	find $(\gamma, f)$ satisfying the Kuhn--Tucker conditions \eqref{eq:kuhnTucker1}
	and \eqref{eq:kuhnTucker2} then this will yield a maximizer
	for the problem	\eqref{eq:abstractInvestmentProblem} in the case
	where the initial budget satisfies
	\begin{equation}
	X_0 = \gamma + s^C_{t_0} \int_0^1 e^{-r \delta t} q^A_{\BS}(s) f(s) \, \ed s.
	\label{eq:abstractProblemBudget}
	\end{equation}
	We remark that the theory of Lagrange multipliers given in \cite{rockafellar} is
	stated in terms of finite dimensional spaces. We may, nevertheless, apply it
	by noting that if $(\gamma, f)$ satisfies the Kuhn--Tucker conditions yet
	is not a maximizer then there must be some direction
	in which we can perturb $(\gamma, f)$ to obtain a higher value for the gain. We may
	now apply the finite dimensional theory to the vector space generated by this perturbation
	to obtain a contradiction.
			
	The result now follows by introducing a variable
	\[
	\eta:=-\lambda \exp( u(\gamma) \delta t)
	\]
	to simplify the equations.
\end{proof}

This completes the proof of Proposition \ref{prop:solutionExponential}.

The outstanding difficulty is proving that an $\eta$ solving $X^\eta=X_0$ exists.
One might attempt to use general duality theory to do this. Theorem 8.3.1 of \cite{luenberger}
ensures that so long as $X_0$ is chosen to satisfy the Slater condition we can guarantee
the existence of a $\lambda$ minimizing the dual problem. However, this theorem does not
guarantee the existence of a maximizer for the primal problem. As a result, even if one knows
the value of $\lambda$ it is still unclear whether a solution to \eqref{eq:kuhnTucker1}
and \eqref{eq:kuhnTucker2} exists. When one introduces the variable $\eta$, this ensures
that $\gamma^\eta$ and $f^\eta$ are well-defined once $\eta$ is known and so the problem
shifts to finding the correct value of $\eta$. We will resolve this issue in 
the cases of interest using a continuity argument in the next section.

\subsection{Numerical approximation of the multi-period problem}

The results of the previous section immediately suggests a numerical method for
solving our investment problems with exponential utility.

We define the minimum acceptable consumption to be
\[
\gamma_{\min} := \inf \{ x \in \R \mid u(x)>-\infty \}.
\]
In addition to the usual assumptions that $u$ is concave and increasing, we assume
\begin{equation}
u^\dagger \text{ is continuous on } (0,\infty)
\label{eq:uDaggerCts}
\end{equation}
and
\begin{equation}
\lim_{p\to 0} u^\dagger(p)=\infty.
\label{eq:uDaggerP0}
\end{equation}
We note that our assumption that $u$ is concave and increasing also ensures that
\begin{equation}
\lim_{p\to \infty} u^\dagger(p)=\gamma_{\min}.
\label{eq:uDaggerPInf}
\end{equation}

\begin{algorithm}
	\label{algo:exponential}	
	Choose a grid of points
	${\cal X}=\{x_1,x_2 \ldots, x_N\}$ on which we will approximate the value function $v_t$.
	We will write $\tilde{v}_t$ for our approximate value function.
	This will be a concave increasing
	piecewise linear function equal to $-\infty$ on $(-\infty,x_1)$, linear on $[x_i,x_{i+1}]$
	and constant on $[x_N,\infty)$. We will simply need
	to store the values ${\tilde{v}}_t(x_i)$ at the grid points.	
	
	To avoid numerical overflow issues we define a function
	$\ell(x):=-\log(-x)$ and store the values $\ell(\tilde{v}_t(x_i))$
	at each grid point rather than storing $\tilde{v}_t(x_i)$ itself.
	\begin{enumerate}[(i)]
		\item Choose the values at the final time point $T-\delta t$ by
		\[
		\tilde{v}_{T-\delta t}( x_i ) := v_{T-\delta t}(x_i) = -\exp(-u(x_i) \delta t).
		\]
		Or equivalently
		\[
		\ell(\tilde{v}_{T-\delta t}( x_i )) = \ell(v_{T-\delta t}(x_i)) = u(x_i) \delta t.
		\]
		\item Suppose that $\tilde{v}_t$ is known. Set
		$\tilde{v}_{t-\delta t}(x_i)$ to be the solution of \eqref{eq:abstractInvestmentProblem}
		with $v_{t_1}=\tilde{v}_t$ and initial budget $x_i$. We describe in detail how
		to solve this problem in Proposition \ref{prop:logFormaule} below.
	\end{enumerate}
\end{algorithm}

Since $v_{T-\delta t}$ is concave and increasing and
$\tilde{v}_{T-\delta t}$ is piecewise linear $\tilde{v}_{T-\delta t}(x)\leq v_{T-\delta t}(x)$. Let $\hat{v}_{t}(x)$
be defined to be the solution of \eqref{eq:abstractInvestmentProblem}
with $v_{t_1}=\tilde{v}_t$ and initial budget $x$. We see
that $\tilde{v}_{T-\delta t}(x)\leq \hat{v}_{T-\delta t}(x)
\leq v_{T-\delta t}(x)$.

Let
${\cal X}_1 \subseteq {\cal X}_2 \subseteq {\cal X}_3 \ldots $ 
be an increasing sequence of grids
with ${\cal X}_\infty:=\cup_{j=1}^\infty {\cal X}_i$ being dense
in $(0,\infty)$. Write $\tilde{v}^j_t$
for the approximations with respect to ${\cal X}_i$.
We see by repeating the argument above that $\tilde{v}^j_t(x)\leq \tilde{v}^j_t(x)$ at all points $x \in (0,\infty)$.
Hence we may define
\[
\tilde{v}_t(x)=\lim_{j\to \infty} \tilde{v}^j_t.
\]

\begin{theorem}[Convergence of Algorithm \ref{algo:exponential}]
	Define
	\[
	X_{\min,t} = \sup\{ x \mid v_t(x) = -\infty \}.
	\]
	For $x > X_{\min,t}$ we have
	\[
	\tilde{v}_t(x) = v_t(x).
	\]
\end{theorem}
\begin{proof}
	Let ${\cal V}$ denote the space of concave, increasing functions $v(x)$ which satisfy $v(x)=-\infty$ for $x<0$ and where $v(x)$ is bounded above by $0$.
	For two adjacent times $t_0, t_1=t_0+\delta t$ in our grid 
	we define a solution function $\phi_{t_0,t_1,X_0}:{\cal V} \to \R$ by
	setting
	\[
	\phi_{t_0,t_1,X_0}(v_{t_1})
	\]
	to equal the supremum in \eqref{eq:exponentialDynamicProgramming}. By composing these
	solution functions in the obvious way, we obtain a solution function $\phi_{t_0,t_1,X_0}$
	for any times in the grid with $t_0\leq t_1$.
	
	We define a corresponding minimum budget as follows:
	\[
	X_{\min,t_0,t_1}(v) = \sup\{ x \mid \phi_{t_0,t_1,x}(v) = \infty \}.
	\]
	
	Let $t_0, t_1$ be adjacent times in the grid.
	Given $v \in {\cal V}$ with $\phi_{t_0,t_1,X_0}(v)$ finite,
	let $(\gamma_{t_0}, \alpha) \in {\cal A}_{X_0,t_0,t_1}$
	be a maximizing strategy for the problem \eqref{eq:exponentialDynamicProgramming} with $v_{t_1}=v$.
	Suppose $w \in {\cal V}$. We have
	\begin{multline*}
	\Big|
	\exp\left(-u(\gamma_{t_0}) \delta t \right)
	\left( -1 + s_{t_0} 
	\E\left( 1+v(X^{(\gamma_{t_0},\alpha)}_1) \right)
	\right)
	- \\
	\exp\left(-u(\gamma_{t_0}) \delta t \right)
	\left( -1 + s_{t_0} 
	\E\left( 1+w(X^{(\gamma_{t_0},\alpha)}_1) \right)
	\right)
	\Big| \\
	\leq C \exp(-u_{\gamma_{t_0}}) \|v-w\|_\infty.
	\end{multline*}
	Hence for any $\epsilon>0$ we can find $\delta_1>0$ such that $\|v-w\|_\infty<\delta_1$ implies
	\[
	\phi_{t_0,t_1,X_0}(w) \geq \phi_{t_0,t_1,X_0}(v)-\epsilon.
	\]
	We have shown $\phi_{t_0,t_1,X_0}$ is lower semi-continuous in the $\sup$ norm for adjacent
	times $t_0$ and $t_1$. It follows that $\phi_{t_0,t_1,X_0}$ is lower semi-continuous for all $t_0<t_1$.
	
	Given $v \in {\cal V}$ and $h \in \R$, define the translation
	\[
	v_h(x)=\begin{cases}
	v(x-h) & x- h \geq 0 \\
	-\infty & x - h < 0.
	\end{cases} = \min \{ v(x-h), (\sup v) \id_{x-h<0} \}
	\]
	Define $f_{t_0,t_1,v}(h)=\phi_{t_0,t_1,X_0}(v_h)$. The function $v(x,h)=v_h(x)$
	is concave.
	Hence $f_{t_0,t_1,v}$
	is concave as a function of $h$. If $X_0>X_{\min,t_0,t_1}(v)$ then 
	$0 \in \ri f_{t_0,t_1,v}$, where $\ri f$ denotes the relative interior of $f$. Hence
	$f_{t_0,t_1,v}$ is continuous in $h$ at $0$. 
	
	Combining this with the lower semi-continuity result, we see that if $X_0>X_{\min,t_0,t_1}(v)$
	then given $\epsilon>0$, we can find $\delta_1>0$ and $\delta_2>0$ such that
	\[
	\phi_{t_0,t_1,x_0}(v_{\delta_1}(x) - \delta_2) \geq 
	\phi_{t_0,t_1,x_0}(v) - \epsilon.
	\]
	Let us write $v_\epsilon(x)$ for the function $v_{\delta 1}(x)-\delta_2$. Given a function $f$
	let us write $\Gamma_f$ for the {\em hypograph} of $f$, that is to say the set of points on or below the graph. We have $\Gamma_v \supseteq \Gamma_{v_\epsilon}$. For any function $w \in {\cal V}$ satisfying
	$\Gamma_v \supseteq \Gamma_w \supseteq \Gamma_v$ we will have
	\[
	\phi_{t_0,t_1,x_0}(v) \geq \phi_{t_0,t_1,x_0}(w) \geq \phi_{t_0,t_1,x_0}(v_\epsilon)
	\geq 
	\phi_{t_0,t_1,x_0}(v) - \epsilon.
	\]
	since it is clear that $\Gamma_w \supseteq \Gamma_w$ implies $\phi_{t_0,t_1,x_0}(v) \geq \phi_{t_0,t_1,x_0}(w)$. Note that we can always find a piecewise linear approximation
	between $\Gamma_v$ and $\Gamma_{v_\epsilon}$.
	
	Given a value for $\epsilon_0$, we may inductively extend this to a sequence of
	positive ${\epsilon_t}$ for $t \in {\cal T}$
	such that if our approximation $\tilde{v}_t$ satisfies
	$\Gamma_{v_t} \supseteq \Gamma_{\tilde{v}_t} \Gamma_{(v_t)_{\epsilon_t}}$
	then it will automatically satisfy 
	$\Gamma_{v_{t-\delta t}} \supseteq \Gamma_{\tilde{v}_{t-\delta t}} \supseteq \Gamma_{(v_t)_{\epsilon_{t-\delta t}}}$.
	By choosing a sufficiently fine grid we can ensure this condition is satisfied at time $T-\delta t$.
	By further refinements we may ensure that it is satisfied at all times.
\end{proof}

Let us now describe in full detail how to
solve \eqref{eq:abstractInvestmentProblem}
given that $v_{t_1}$ is of the form used in our algorithm.
In Proposition \ref{prop:logFormaule}, we will give the formulae
necessary to solve the problem on a computer in a format that addresses
numerical overflow issues.  Terms on the left hand side of the equations in the Proposition
should be stored in computer memory and can be computed without overflow issues from the terms on the right. We use infinite values for some terms as a convenient shorthand, terms such as an exponential of $-\infty$ should be interpreted in the obvious way.

To store probability values we define a bijection $L:[0,1]\to \R\cup \{\pm \infty \}$
by
\[
L(u) = \begin{cases}
\log(2u ) & u\leq0.5 \\
-\log(2-2u) & u> 0.5.
\end{cases}
\]
We note that the GNU scientific library contains a function {\tt gsl\_sf\_log\_erfc}
which computes the logarithm of the complementary error function which we can then use to compute $L(\Phi)$.

We define a function
\[
\tilde{u}(y) = \log( u^\dagger (e^y) ).
\]
For the specific functional form
\[
u(x)=\begin{cases}
a (x-x_0)^n + b & x \geq 0 \\
-\infty & \text{otherwise}
\end{cases}
\]
which we will use in our numerical examples, we may compute
$\tilde{u}$ without experiencing
overflow errors using the formulae
\begin{align}
\tilde{u}_0(p)&:= \frac{1}{n - 1}(p - \log(a \, n)), \\
\tilde{u}(y)&= \begin{cases}
\log( e^{\tilde{u}_0(p)} ) & x_0 = 0 \\
\log( e^{\tilde{u}_0(p)} + e^{\log(x_0)}) & x_0 > 0 \\
\log( e^{\tilde{u}_0(p)} - e^{\log(-x_0)} ) & \tilde{u}_0(p) > \log(-x_0) \text{ and } x_0 < 0 \\
-\infty & \tilde{u}_0(p) \leq \log(-x_0) \text{ and } x_0 < 0.
\end{cases}
\end{align}
We note the standard approach to computing the log of sums and differences of exponentials
without overflow issues should be used when evaluating
expressions such as this.

\begin{proposition}
	\label{prop:logFormaule}
	Let $v$ be a concave, non-positive, increasing function
	which is linear between grid points	in ${\cal X}=\{x_1,x_2, \ldots x_N\}$
	with $x_i$ strictly increasing. Suppose also that $v$ is equal
	to $-\infty$ on $(-\infty,x_1)$ and constant on $(x_N,\infty)$.
    Suppose that $u^\dagger$ is continuous and 
	satisfies equations \eqref{eq:uDaggerCts} and \eqref{eq:uDaggerP0}.
	
	Define a decreasing sequence of points $\log( p_i )$ by
	\begin{equation}
	\log( p_i ) = \begin{cases}
	\infty & i=0 \\
	\log (  e^{ - \ell( v(x_{i})) } -e^{-\ell(v(x_{i+1}))} ) - \log(x_{i+1}-x_{i}) & 0<i<N \\
	-\infty & i=N.
	\end{cases}
	\label{eq:defLogDualPoints}
	\end{equation}
	For a given value of $\log \eta$, define $L(U^\eta_i)$ and $L(Q^\eta_i)$ for $0<i<N$ by
	\begin{align}
	L(U^\eta_i) &= L \left( \Phi\left( -\frac{1}{2}M + \frac{1}{M} \left( \log (\eta ) - r \delta t - (1-C) \log(  s_{t_0} ) - \log( p_i ) \right) \right) \right),
	\label{eq:logUIExplicit} \\
	L(Q^\eta_i) &= L \left( \Phi\left( \frac{1}{2}M + \frac{1}{M} \left( \log (\eta ) - r \delta t - (1-C) \log(  s_{t_0} ) - \log( p_i ) \right) \right) \right).
	\label{eq:logQIExplicit}
	\end{align}
	Define $L(U^\eta_0)=L(Q^\eta_0)=-\infty$ and $L(U^\eta_N)=L(Q^\eta_N)=\infty$.
	We may then define the quantity $A^\eta$ by
	\[
	A^\eta
	=\log\left(
	e^{\log(1-s_{t_0})}+ \sum_{i=1}^N e^{ \log(s_{t_0}) + \log(-v(x_i)) + \log\left(e^{\log U^\eta_{i}}-e^{\log U^\eta_{i-1}} \right) }\right).
	\]
	We then have that
	\begin{equation}
	\log (\gamma^\eta) = \tilde{u}( \log(\eta) - \log(\delta t) - A^\eta )
	\label{eq:logGammaEtaExplicit}
	\end{equation}
	where $\gamma^\eta$ is as defined in \eqref{eq:defGammaEta}.
	We have
	\begin{equation}
	\log(X^\eta)= \log\left(
	e^{\log( \gamma^\eta)} +  \sum_{i=1}^N e^{C \log(s_{t_0}) -r \delta t + \log(x_i)+\log\left(e^{\log Q^\eta_i}-e^{\log Q^\eta_{i-1}} \right)}
	\right)
	\label{eq:logxEtaExplicit}
	\end{equation}
	and $X^\eta$ depends continuously upon $\eta$.
	If $X_0 > s^C_{t_0} e^{-r \delta t} x_1 + \gamma_{\min}$,
	we may find the value of $\eta_{X_0}$ by finding
	$\log(\eta)$ such that $\log(X^\eta)=\log(X_0)$.
	We then have
	\begin{align}
	\ell( v(t_0,X^0))=
	u(\gamma^\eta)  \delta t - A^\eta.
	\label{eq:logV0Explicit}
	\end{align}
	If $X_0 < s^C_{t_0} e^{-r \delta t} x_1$, the maximum in \eqref{eq:abstractInvestmentProblem} is $-\infty$ which is achieved by the negative consumption $\gamma = X_0 - s^C_{t_0} e^{-r \delta t} x_1$.
\end{proposition}
\begin{proof}
	Corresponding to \eqref{eq:defLogDualPoints} we have a decreasing sequence of points $p_i$ given
	by
	\begin{equation}
	p_i = \begin{cases}
	\infty & i=0 \\
	\frac{v(x_{i+1})-v(x_{i})}{x_{i+1}-x_{i}} & 0<i<N \\
	0 & i=N.
	\end{cases}
	\label{eq:defDualPoints}
	\end{equation}
	We will then have
	\[
	v^\dagger(p) =
	\sum_{i=1}^N x_i \id_{[p_{i},p_{i-1})}(p).
	\]
	From \eqref{eq:defFEta}
	\[
	f^\eta(u) = \sum_{i=1}^N x_i \id_{[p_{i},p_{i-1})}\left( 
	\eta  s^{C-1}_{t_0} e^{-r \delta t} q^A_{\BS}(u) \right).
	\]
	Hence we will be able to deduce that
	\begin{equation}
	f^\eta(U) = \sum_{i=1}^N x_i \id_{(U^\eta_{i-1},U^\eta_i]}\left( 
	U \right)
	\label{eq:fEtaAsSum}
	\end{equation}
	if we can show \eqref{eq:logUIExplicit} ensures that
	\begin{equation}
	\eta  s^{C-1}_{t_0} e^{-r \delta t} q^A_{\BS}(U^\eta_i) = p_i.
	\label{eq:uIImplicit}
	\end{equation}
	Writing $\phi$ for the pdf of the standard normal we compute
	\begin{align*}
	q^A_{\BS}(u) &=
	\frac{\phi( M + \Phi^{-1}(u))}{\phi( \Phi^{-1}(u))} \\
	&= \exp\left( \frac{1}{2}( \Phi^{-1}(u)^2 - (M + \Phi^{-1}(u))^2 )\right) \\
	&= \exp\left( -\frac{1}{2}M^2 - M \Phi^{-1}(u) \right).
	\end{align*}
	Hence equation \eqref{eq:uIImplicit} is equivalent to
	\begin{equation}
	U^\eta_i = \Phi\left( -\frac{1}{2}M - \frac{1}{M} \log \left( \frac{1}{\eta} s^{1-C}_{t_0} e^{r \delta t} p_i \right) \right).
	\label{eq:uIExplicit}
	\end{equation}
	for $0<i<N$, which will hold due to our definition \eqref{eq:logUIExplicit}.
	From \eqref{eq:defGammaEta} and \eqref{eq:fEtaAsSum} we have
	\begin{align}
	\gamma^\eta  &= u^\dagger \left( -\frac{\eta}{\delta t}
	\left(
	-1
	+ s_{t_0} \int_0^1 \left(1+v\left( \sum_{i=1}^N x_i \id_{(U^\eta_{i-1},U^\eta_i]}(s) \right) \right) \, \ed s \right)^{-1}
	\right) \nonumber \\
	&= u^\dagger \left( -\frac{\eta}{\delta t}
	\left(
	-1
	+ s_{t_0} \int_0^1 \left(1+ \sum_{i=1}^N v(x_i) \id_{(U^\eta_{i-1},U^\eta_i]}(s) \right) \, \ed s \right)^{-1}
	\right) \nonumber \\
	&= u^\dagger \left( -\frac{\eta}{\delta t}
	\left(
	-1
	+ s_{t_0} \left(1+ \sum_{i=1}^N v(x_i) (U^\eta_i-U^\eta_{i-1})\right) \right)^{-1}
	\right).
	\end{align}
	Equation \eqref{eq:logGammaEtaExplicit} follows immediately.
	
	Use \eqref{eq:defXeta} and \eqref{eq:defQ} to see that
	\begin{align}
	X^\eta &= \gamma^\eta + s^C_{t_0} \sum_{i=1}^N \int_{U^\eta_{i-1}}^{U^\eta_i} e^{-r \delta t} q^A_{\BS}(s) x_i \, \ed s \nonumber \\
	&= \gamma^\eta + s^C_{t_0} \sum_{i=1}^N e^{-r \delta t} x_i (Q(U^\eta_i)-Q(U^\eta_{i-1})) \nonumber \\
	&= \gamma^\eta + s^C_{t_0} \sum_{i=1}^N e^{-r \delta t} x_i (Q^\eta_i-Q^\eta_{i-1}))
	\label{eq:xEtaExplicit}
	\end{align}
	The last line follows directly from our definitions of $Q$, $U^\eta_i$ and $Q^\eta_i$.
	We now see that equation \eqref{eq:xEtaExplicit} is equivalent to
	\eqref{eq:logxEtaExplicit}.
	
	Our explicit formula, \eqref{eq:uIExplicit}, for $\gamma^\eta$ shows that it depends continuously $\eta$	given the assumption \eqref{eq:uDaggerCts}.
	It then follows from equation \eqref{eq:xEtaExplicit} that $X^\eta$ depends continuously on $\eta$.
	Lemmas \eqref{lemma:xForSmallEta} and \eqref{lemma:xForLargeEta} below then establish
	that we can solve for $\eta$ in $X^\eta = X_0$ under the conditions of the proposition.
	
	The value function
	is then given by
	\begin{align*}
	v(t_0,X^0)&=
	\exp( -u(\gamma^\eta)  \delta t ) \left(
	-1
	+ s_{t_0} \int_0^1 (1+v( 
	\sum_{i=1}^N x_i \id_{(U^\eta_{i-1},U^\eta_i]}\left( 
	s \right)
	)) \, \ed s \right) \\
	&=
	\exp( -u(\gamma^\eta)  \delta t ) \left(
	-1
	+ s_{t_0}(1+ 
	\sum_{i=1}^N v(x_i) (U^\eta_{i}-U^\eta_{i-1}) ) \right)
	\end{align*}
	and so \eqref{eq:logV0Explicit} also follows.
\end{proof}

\begin{lemma}
\label{lemma:xForSmallEta}	
Under the assumptions of Proposition \ref{prop:logFormaule},
\[
\lim_{\eta \to 0} X^\eta = \infty.
\]	
\end{lemma}
\begin{proof}
Our assumptions on $v$ ensure that
\[
-1 +s_{t_0} \int_0^1 (1 + v(f^\eta(s))) \, \ed s \leq -1+s_{t_0} <0.
\]	
Hence
\[
0 > \left( -1 +s_{t_0} \int_0^1 (1 + v(f^\eta(s))) \right)^{-1} < \frac{1}{-1+s_{t_0}}.
\]
It now follows from our equation \eqref{eq:uDaggerP0} coupled with equation \eqref{eq:defGammaEta}
that
\[
\lim_{\eta \to 0} \gamma^\eta = \infty.
\]
The result now follows from \eqref{eq:defXeta}.
\end{proof}
\begin{lemma}
Under the assumptions of Proposition \ref{prop:logFormaule}, 
\[
\lim_{\eta \to \infty} \gamma^\eta = \gamma_{\min}.
\]	
\label{lemma:limGammaEta}
\end{lemma}	
\begin{proof}
Our assumptions on $v$ ensure that
\[
\left(
-1
+ s_{t_0} \int_0^1 (1+v(f^\eta(s))) \, \ed s \right)^{-1}
\]
is bounded. Hence using the expression \eqref{eq:defGammaEta}
combined with assumption \eqref{eq:uDaggerPInf} we find $\gamma^\eta \to 0$
as $\eta \to \infty$.
\end{proof}
\begin{lemma}
Under the assumptions of Proposition \ref{prop:logFormaule}, 
\[
\lim_{\eta \to \infty} X^\eta = \gamma_{\min} + s_{t_0}^C e^{-r \delta t } x_1 .
\]		
\label{lemma:xForLargeEta}	
\end{lemma}
\begin{proof}
Define 
\[
p^* = \inf \partial v(x_1).
\]	
For $\eta >0$, define
\begin{equation}
s^*_\eta = q^A_{\BS}(p^* \eta^{-1} s^{1-C}_{t_0} e^{rt}),
\label{def:sStarEta}
\end{equation}
which ensures that
\begin{equation}
s \geq s^* \iff \eta s^{C-1}_{t_0} e^{rt} q^A_{\BS}(s) < p^{\star}.
\label{eq:sStarCondition}
\end{equation}
We compute
\begin{align}
\int_0^1 q^A_{\BS}(s) f^\eta(s) \ed s 
&= \int_0^1 q^A_{\BS}(s) v^\dagger( \eta s^{C-1}_{t_0} e^{-rt} q^A_{\BS}(s) ) \, \ed s
\nonumber \\ 
&= 
\int_0^{s^*_\eta} q^A_{\BS}(s) v^\dagger( \eta s^{C-1}_{t_0} e^{-rt} q^A_{\BS}(s) ) \, \ed s
\nonumber \\
&\qquad + \frac{1}{\eta s^{C-1}_{t_0} e^{-rt}} \int_{s^*_\eta}^1 \eta s^{C-1}_{t_0} e^{-rt} \, q^A_{\BS}(s) v^\dagger( \eta s^{C-1}_{t_0} e^{-rt} q^A_{\BS}(s) ) \, \ed s
\nonumber \\
&\leq 
\int_0^{s^*_\eta} q^A_{\BS}(s) x_1 \, \ed s
\nonumber \\
&\qquad + \frac{1}{\eta s^{C-1}_{t_0} e^{-rt}} \int_{s^*_\eta}^1 \eta s^{C-1}_{t_0} e^{-rt} \, q^A_{\BS}(s) v^\dagger( \eta s^{C-1}_{t_0} e^{-rt} q^A_{\BS}(s) ) \, \ed s 
\label{eq:xEtaInequality1} \\
\end{align}

We note that $p \in \partial v( v^\dagger(p))$.
By the definition of the subdifferential at $v^\dagger(p)$
\[
v(x) \leq v(v^\dagger(p)) + p( x - v^\dagger(p)).
\]
Rearranging yields
\[
p v^\dagger(p) \leq p x + v(v^\dagger(p)) - v(x).
\]
Using the fact $v$ is increasing and substituting $x_1$ for $x$ we find that for all $p$
\[
p v^\dagger(p) \leq p x_1 + v(x_N) - v(x_1).
\]
Using this inequality in \eqref{eq:xEtaInequality1}  we find
\begin{align}
\int_0^1 q^A_{\BS}(s) f^\eta(s) \ed s 
&\leq 
\int_0^{s^*_\eta} q^A_{\BS}(s) x_1 \, \ed s 
\nonumber \\
&\qquad + 
\frac{1}{\eta s^{C-1}_{t_0} e^{-rt}}
\int_{s^*_\eta}^1 ( \eta s^{C-1}_{t_0} e^{-rt} \, q^A_{\BS}(s) x_1 + v(x_N) - v(x_1)) \, \ed s
\nonumber \\
&\leq 
\int_0^{s^*_\eta} q^A_{\BS}(s) x_1 \, \ed s
\nonumber \\
&\qquad + 
\frac{1}{\eta s^{C-1}_{t_0} e^{-rt}}
\int_{s^*_\eta}^1 ( p^* x_1 + v(x_N) - v(x_1)) \, \ed s.
\label{eq:xEtaInequality2}
\end{align}
by \eqref{eq:sStarCondition}. From \eqref{def:sStarEta}
\[
\lim_{\eta \to \infty} s^*_\eta = 1.
\]
We may therefore take the limit of the inequality \eqref{eq:xEtaInequality2}
to find
\[
\liminf_{\eta>0}
\int_0^1 q^A_{\BS}(s) f^\eta(s) \ed s  \leq x_1.
\]
Using this, Lemma \ref{lemma:limGammaEta} and the definition of $X^\eta$ in
equation \eqref{eq:defXeta} we find
\[
\liminf_{\eta>0} X^\eta \leq \gamma_{\min} + s_{t_0}^C e^{-r \delta t} x_1.
\]
From \eqref{eq:defGammaEta} and \eqref{eq:defXeta} one sees that, on the other hand, 
for all $\eta>0$ we have
\[
X^\eta \geq \gamma_{\min} + s_{t_0}^C e^{-r \delta t} x_1.
\]
The result follows.
\end{proof}

\begin{remark}
	We note that that if we follow the optimal investment
	strategy at time $t$, then the optimal investment strategy will result
	in a wealth at time $t+\delta t$ which takes values in the grid $\{x_1,\ldots,x_n\}$.
	We may then approximate the value function on the space-time grid
	$\{x_1,\ldots x_n \} \times \{0,\delta t, 2 \delta t, \ldots, T\}$.
	One can then obtain a simulation of the optimal strategy by first simulating 
	the stock price on the time grid and then computing the corresponding dynamics
	of $x_t$ in the grid $\{x_1,\ldots x_n \}$ using this approximation to the value
	function. Since the wealth process never leaves a fixed space-time grid, we can use
	the same approximation of the value function for all the scenarios.
\end{remark}

\begin{remark}
	When implementing this algorithm we notice that many values
	of $U_i^\eta$ will be extremely close to either 0 or 1, and so including
	these terms will have a negligible effect on the values of the sums
	in the equations \eqref{eq:logGammaEtaExplicit}, \eqref{eq:logxEtaExplicit}.
	Financially this is equivalent to ignoring extreme events of
	very low probability where the
	$\P$ and $\Q$ disagree by a large amount. Since our payoff functions $f$
	take values in ${\cal X}$, and so are bounded and positive, ignoring
	these extreme events will have no material impact upon either the price or the expected utility.
	The value we chose in our numerical calculations was $\epsilon=10^{-10} \max{|v(x_i)|}^{-1}$.
	
	This can be used to speed up the algorithm. When
	calculating $X^\eta$, choose some small $\epsilon$ and define
	\begin{align*}	
	i_{\min}&:=\max \{1\} \cup \{ i \mid U_i<\epsilon\} \\
	i_{\max}&:=\min  \{N\} \cup \{ i \mid U_i>1-\epsilon\}.
	\end{align*}
	To compute these values and the values of $U_i$,
	first use the method of bisection to find some
	$i^*$ where $\epsilon<U_{i^*}<1-\epsilon$. Then compute the values of $U_i$
	from $i^*$ down to $i_{\min}$, stopping when $U_i<\epsilon$. Similarly compute the values of $U_i$ from $i^*$ up to $i_{\max}$, stopping when $U_i>1-\epsilon$.
	No other values of $U_i$ outside the range $i_{\min-1}\leq i \leq i_{\max}$ are then needed in the computation of $X^\eta$.
	
	When computing the values of the sums in \eqref{eq:logGammaEtaExplicit}, \eqref{eq:logxEtaExplicit} use indices running from $i_{\min}$ to $i_{\max}$ rather than form $1$ to $n$. 
\end{remark}

\section{Numerical Tests}
\label{sec:numericalTests}

We refer to the paper \cite{ab-main} for detailed numerical
results for a model calibrated with realistic model parameters.

In this
paper we will restrict ourselves to showing how our algorithm can
be tested using by making a comparison with the analytic results
of Section \ref{sec:mattress} and \cite{ab-ez}.

We begin with a comparison with the results of \cite{ab-ez}.

If we choose $u(x)$ to represent an individual with 
low satisfaction-risk-aversion, one would expect the consumption in the exponential Kihlstrom--Mirman
case to be very close to the curves for von Neumann--Morgernstern preferences.

In Figure \ref{fig:exp-optimal-consumption-collective} we have plotted a fan diagram
for the optimal consumption
in the individual case. The model and algorithm parameters used are given in 
Table \ref{table:parameters}.

The smoother lines in this fan diagram shows empirical consumption percentiles $(1,5,50,95,99)$ obtained for a sample of $10^5$ scenarios where the optimal strategy was followed. The jagged
line shows a single illustrative scenario.

Figure \ref{fig:exp-optimal-consumption-vnm} shows the consumption when computed 
analytically for an individual with von Neumann--Morgernstern preferences.
As expected, the differences from  Figure \ref{fig:exp-optimal-consumption-collective} are small.

\begin{table}[htbp]
\begin{center}
\begin{tabular}{ll} \toprule
Parameter & Value \\	\midrule
$\delta_t$ & $0.02$ \\	
$r$ & $0.02$ \\	
$\mu$ & $0.05$ \\	
$\sigma$ & $0.15$ \\	
$X_0$ & $65$ \\	
$u(x)$ & $-\frac{1}{10} x^{-2}$ \\	
$w(x)$ & $-\exp(-x)$ \\	
Number of grid points & $1001$ \\	
Grid range & $[x_1,x_n]=[0,195]$ \\	
$p_t$ & CMI\_2018\_F [1.5\%] (see \cite{cmi2018}) \\ \bottomrule
\end{tabular}	
\end{center}	
\label{table:parameters}
\caption{Parameter values used in the numerical examples}
\end{table}

The fan diagram Figure \ref{fig:exp-wealth-fan-collective} shows $(1,5,50,95,99)-th$ percentiles for the  total wealth $X_t$ at each time. In particular this shows that the wealth was unlikely to reach the top boundary. When using the algorithm one should run a check such as this as otherwise one would expect
the value to be significantly suboptimal.

\begin{figure}[p]
	\centering
	\includegraphics{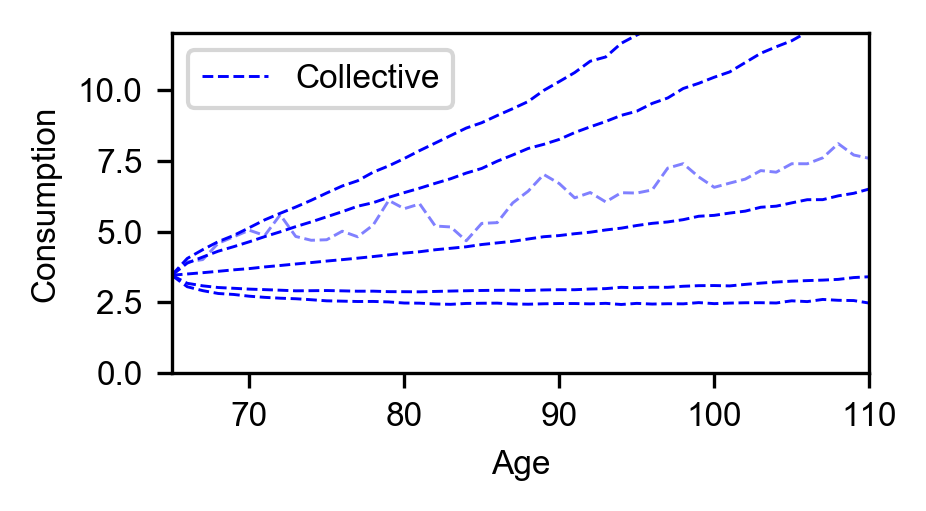}
	\caption{Fan diagram of the optimal consumption for an infinite collective.}
	\label{fig:exp-optimal-consumption-collective}
\end{figure}

\begin{figure}[p]
	\centering
	\includegraphics{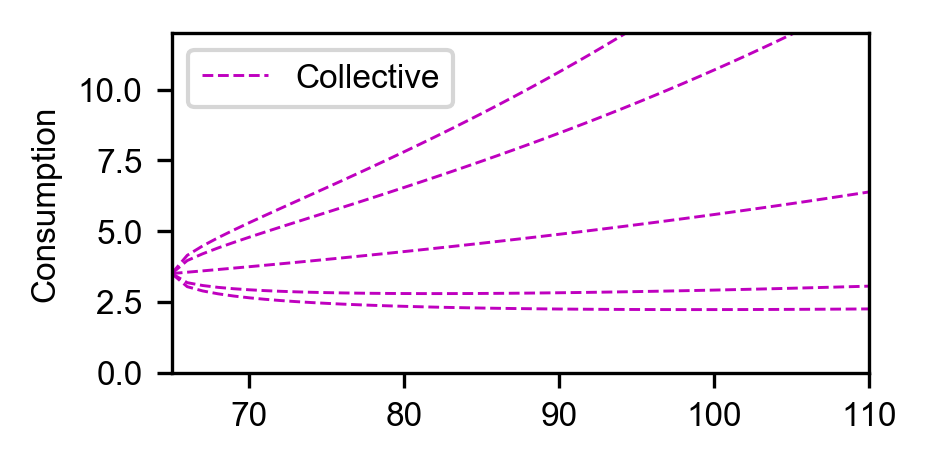}
	\caption{Fan diagram of the optimal consumption for an infinite collective
		 using von Neumann--Morgernstern preferences.}
	\label{fig:exp-optimal-consumption-vnm}
\end{figure}

\begin{figure}[p]
	\centering
	\includegraphics{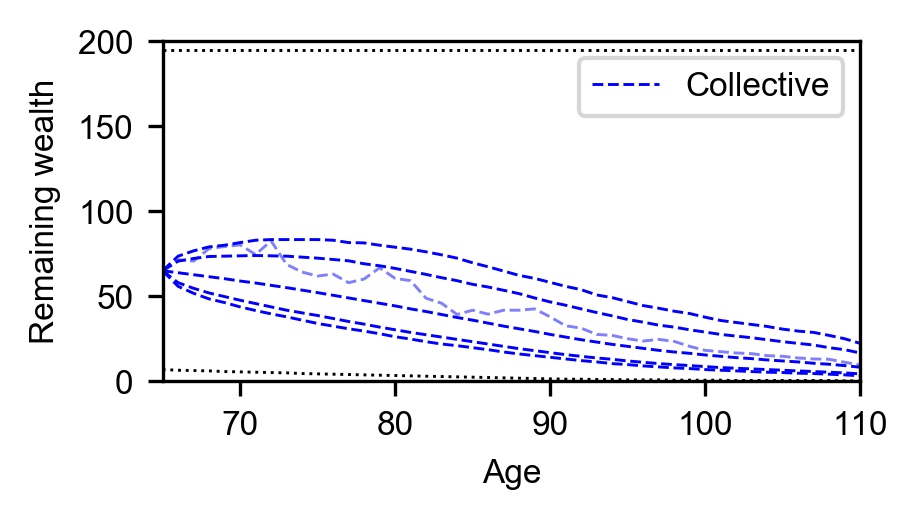}
	\caption{Fan diagram of the remaining wealth per individual at each time. The dotted lines 
		indicate the domain of the approximation to $v_t$ at each time point.}
	\label{fig:exp-wealth-fan-collective}
\end{figure}

\medskip

Next we compare the analytic results of Section \ref{sec:mattress} with the numerical results.
We consider the continuous time optimal investment of an individual with initial wealth $B=3$,
exponential mortality with $\lambda=0.025$
and preferences of the form with $a=0.05$, $b=-0.01$, $k=0.5$. In a market
with $\mu=r=0$ we may compute their optimal consumption analytically. If we instead
consider a market with $r=0$, $\mu=0.0001$ and $\sigma=1.0$ we expect that the optimal
strategy in this case will be very close to that when $\mu=0.0001$. Given our choice
of $\lambda$, we expect that a discrete time approximation to the continuous time 
problem with $\delta t=1$ and $T=200$ should also be accurate. 

In Figure \ref{fig:exp-consumption_analytic} we plot the analytic curve and a fan
diagram for the discrete time approximation with the given market parameters. Again,
we used a grid size of $1000$ for our computation. As
expected the analytic curve and the fan diagram are very close.

\medskip

These tests give us a good degree of confidence in both our analytic formulae
and numerical implementation. Other tests we have performed include: confirming that 
if one simulates in price paths in the $\Q$ measure, the discounted expected consumption
matches the initial budget; checking that the value function is concave and increasing
in the budget at all times $t$; comparison with brute-force computation of the deterministic
optimal investment strategy in the case where $\mu=r$.

\begin{figure}[p]
	\centering
	\includegraphics{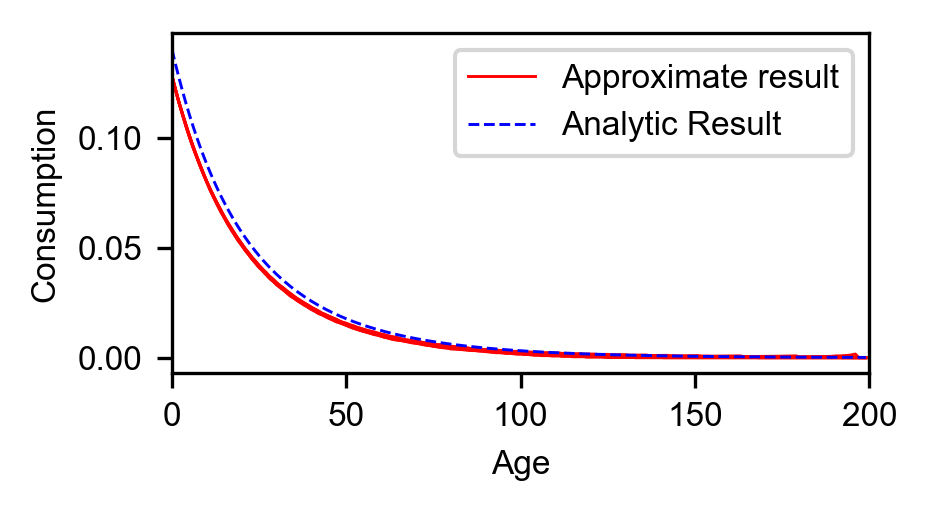}
	\caption{The optimal consumption for an example problem computed analytically
	and using a numerical approximation}
	\label{fig:exp-consumption_analytic}
\end{figure}

\section{Convergence of $v_n$ as $n\to \infty$}
\label{sec:exponentialConvergence}

In this section we will prove that $v_n \to v_\infty$ as $n\to \infty$ under mild assumptions.

We will assume that the
function $u$ in \eqref{eq:kihlstromMirmanGain}
is, up to affine transformation, a CRRA utility function.
Without loss of generality we may ignore the scale factor
and so we assume that for some constants $c$, $a<1$, $a \neq 0$
\begin{equation}
u(x) = \begin{cases}
x^a + c & 0<a<1 \\
-x^\alpha + c & a<0.
\end{cases}
\label{eq:crra}
\end{equation}
We do not consider $\log$ utility in order to avoid special cases in proofs, but we 
do not believe extending our proofs to include this case would present any difficulties.

We will assume that there is a limit on the expected $u$ utility
that can be achieved by investing in this market without consumption
\begin{equation}
\sup_{\stackrel{(\gamma,\bm{\alpha}) \in {\cal A}}{\gamma=0}} \E( u(F^{(\gamma,\bm{\alpha})}_T) ) < \infty
\label{eq:finiteUtility}
\end{equation}
where $F^{(\gamma,\bm{\alpha})}_T$ is the fund value at time $T$
of following the investment-consumption strategy $(\gamma,\bm{\alpha})$.
For example, the solution to the Merton problem ensures that this will be the case for the Black--Scholes--Merton market.

We will also assume that there is an upper-bound on longevity. We define
\[
T^*= \inf\{t \in {\cal T} \mid \tau < t \text{ almost surely} \}.
\]

Our assumptions on the market allow us to show that consumption
in the final moments of life have a negligible effect on overall utility.
This is the interpretation of the next two lemmas.
\begin{lemma}
	In a market  satisfies \eqref{eq:finiteUtility}, for any $(\gamma, \bm{\alpha}) \in {\cal A}$,
	\[
	\lim_{\epsilon\to 0} \E\left( \int_{T^*-\epsilon}^{T^*} u(\gamma_t) \, \DT \right) \leq 0.
	\]
	for a CRRA utility function of the form \eqref{eq:crra}.
	\label{lemma:utilityAtDeathInsignificant}
\end{lemma}
\begin{proof}
	The result is trivial if ${\cal T}$ is discrete, so we assume ${\cal T}$
	is an interval.
	
	The constant $c$ in the expression  \eqref{eq:crra}
	for $u(x)$ has no effect on the value
	of this limit, so we may assume it is zero.
	
	Given $(\gamma, \bm{\alpha}) \in {\cal A}$ define $\tilde{\bm{\alpha}}$ to be
	the investment strategy given by the portfolio given by $\bm{\alpha}$
	together with an additional investment of $\gamma_t$ in the risk-free
	asset at time $t$. Let $F^{\tilde{\bm{\alpha}}}_t$
	denote the fund value of the strategy $(0,\tilde{\bm{\alpha}})$ at time $t$. Since $(0,\tilde{\bm{\alpha}}) \in {\cal A}$ we have
	from the assumption \eqref{eq:finiteUtility} that
	\[
	\E\left( u \left( \int_0^T e^{r(T-t)} \gamma_t \, \DT \right) \right) 
	\leq \E( F^{(0,\tilde{\bm{\alpha}})}_T ) \leq C_1
	\]
	for some finite $C_1$. Hence
	\[
	\E\left( u \left( \int_{T^*-\epsilon}^{T^*} \gamma_t \, \DT \right) \right) \leq C_2
	\]
	for some finite $C_2$.
	
	By the positive homogeneity of $u$, we now have
	\begin{equation}
	\E\left( |u(\epsilon)| u\left( \frac{1}{\epsilon} \int_{T^*-\epsilon}^{T^*} \gamma_t \, \DT
	\right) \right) \leq C_2.
	\label{eq:jensenable}
	\end{equation}
	The measure of $[T^*-\epsilon,T^*]$ with respect to $\frac{1}{\epsilon} \DT$
	is one, so we may think of the
	term 
	\[
	\frac{1}{\epsilon} \int_{T^*-\epsilon}^{T^*} \gamma_t \, \DT
	\]
	as an expectation. So we may apply Jensen's inequality to
	\eqref{eq:jensenable}
	to obtain
	\[
	\E\left( |u(\epsilon)| \frac{1}{\epsilon} \int_{T^*-\epsilon}^{T^*} u(\gamma_t) \, \DT \right) \leq C_2.
	\]
	Hence
	\[
	\E\left( \int_{T^*-\epsilon}^{T^*} u(\gamma_t) \, \DT \right) \leq \frac{\epsilon}{|u(\epsilon)|} C_2 = |\epsilon|^{1-a} C_2.
	\]
\end{proof}

\begin{lemma}
	Let $(\gamma, \bm{\alpha}) \in {\cal A}$ be a strategy which satisfies
	$\gamma_t \geq \delta$ for some $\delta>0$ for all $t \in (0,T^*)$.
	Define a strategy $((\gamma_\epsilon),\bm{\alpha})$ by following the
	investment-consumption strategy $(\gamma, \bm{\alpha})$ up to time $T^*-\epsilon$ but then reduce consumption to $\delta$.
	Under the assumptions of Lemma \ref{lemma:utilityAtDeathInsignificant}
	\[
	\liminf_{\epsilon\to 0} \E\left( -\exp\left( 
	-\int_0^\tau u((\gamma_\epsilon)_t) \DT \right) \right)
	\geq 
	\E\left( -\exp\left( 
	-\int_0^\tau u(\gamma_t) \DT \right) \right).
	\]
	\label{lemma:gammaEpsilon}
\end{lemma}
\begin{proof}
	Let us write $T_\epsilon$ for the term we wish to bound. Choose $\lambda \in (0,1)$,
	then
	\begin{align*}	
	T_\epsilon&:=\E\left( -\exp\left( 
	-\int_0^\tau u((\gamma_\epsilon)_t) \DT \right) \right)	\\
	&=
	\E\left( -\exp\left( 
	- \int_0^\tau u(\gamma_t) \DT 
	- \int_{T-\epsilon}^{\max\{\tau,T^*-\epsilon\}}
	[u(\delta) - u(\gamma_t)] \DT
	\right) \right)	 \\
	&\geq
	\lambda \E\left( -\exp\left( 
	-\frac{1}{\lambda} \int_0^\tau u(\gamma_t) \DT \right) \right) \\
	&\quad +(1-\lambda) \E\left( -\exp\left( 
	-\frac{1}{1-\lambda} \int_{T-\epsilon}^{\max\{\tau,T^*-\epsilon\}}
	[u(\delta) - u(\gamma_t)] \DT
	\right) \right)
	\end{align*}	
	by the concavity of $-\exp$. Then by Jensen's equality
	\begin{align*}	
	T_\epsilon&\geq
	\lambda \E\left( -\exp\left( 
	-\frac{1}{\lambda} \int_0^\tau u(\gamma_t) \DT \right) \right) \\
	&\quad -(1-\lambda) \exp\left( \E \left( 
	-\frac{1}{1-\lambda} \int_{T-\epsilon}^{\max\{\tau,T^*-\epsilon\}}
	[u(\delta) - u(\gamma_t)] \DT
	\right) \right).	
	\end{align*}	
	Hence by Lemma \ref{lemma:utilityAtDeathInsignificant}
	\begin{align*}	
	\liminf_{\epsilon\to 0} T_\epsilon&\geq
	\lambda \E\left( -\exp\left( 
	-\frac{1}{\lambda} \int_0^\tau u(\gamma_t) \DT \right) \right) 
	-(1-\lambda) \\
	&=
	\lambda \E\left( -\exp\left( 
	-\int_0^\tau u(\gamma_t) \DT \right)^{\frac{1}{\lambda}} \right) 
	-(1-\lambda).
	\end{align*}	
	Taking the limit as $\lambda \to 1$ and using the continuity properties of
	$L^p$ norms in $p$, we obtain the desired result.
\end{proof}

The next lemma allows us to estimate the probability of a significant
deviation from the expected rate of mortality over the continuum of time points $[0,t_0]$ by estimating the probability of a more significant deviation from the expected rate of mortality over a finite set of time points.
\begin{lemma}
	For a continuous mortality distribution,
	given a time point $0\leq t_0<T^*$, for any $\epsilon \in (0,1)$ there exists
	a finite set of points $t_i \in [0,t_0)$, indexed by $i \in I$ such that
	\[
	\P\left( \forall t \in [0,t_0] \,:\, n_t \leq \left( \frac{1}{1-\epsilon} \right)^2 \E(n_t) \right) 
	\geq 
	\P\left( \forall i \in I \,:\, n_{t_i} \leq \left( \frac{1}{1-\epsilon} \right) \E(n_{t_i}) \right)
	.
	\]
	\label{lemma:finiteTimePoints}
\end{lemma}
\begin{proof}
	We define $t_i$ inductively. If $t_{i-1}=0$, we are done and take the index 
	set $I=\{0,1,\ldots,i-1\}$. Otherwise define
	\[
	t_i = \inf \left\{ t \mid t=0 \text{ or } \E(n_{t}) \leq \frac{1}{1-\epsilon} \E(n_{t_{i-1}}) \right\}.
	\]
	If $t_i\neq 0$ we see $\E(n_{t_{i-1}}) \geq \frac{1}{1-\epsilon} \E(n_{t_i})$,
	so for sufficiently large $i$ we must have $t_i=0$. Hence the index set, $I$, is finite.
	
	Given $t \in [0,t_0]$ we can find $i \in I$ with $t_i \leq t \leq t_{i-1}$.
	Suppose 
	\[
	n_t > \left( \frac{1}{1-\epsilon}\right)^2 \E( n_t)	
	\]
	then we have
	\[
	n_{t} > \left( \frac{1}{1-\epsilon}\right)^2 \E( n_{t_i})
	\geq \left( \frac{1}{1-\epsilon}\right)^2 \E( n_{t_{i-1}})
	\geq \left( \frac{1}{1-\epsilon}\right) \E( n_{t}).
	\]
\end{proof}

Our final lemma is a basic continuity result for concave functions.
\begin{lemma}
	Let $C$ be an open convex subset of $\R^n$, let $f$ be a concave
	function on $\R^n$ and let $c^*$ be a point in the boundary $\partial C$,
	then
	\[
	\sup_{c \in C} f(c) \geq  f(c^*).
	\]
	\label{eq:concavityLemma}
\end{lemma}
\begin{proof}
	Suppose for a contradiction that $\sup_{c \in C} f(c) = f(c^*)-\epsilon$
	for some $\epsilon>0$. We can then find a point $c \in C$
	with $f(c)>f(c^*)-2\epsilon$.
	\[
	f(\frac{3}{4}c^* + \frac{1}{4}c)
	\geq
	\frac{3}{4}f(c^*) + \frac{1}{4}f(c)
	\geq f(c^*)-\frac{1}{2}\epsilon.
	\]
	This then implies that $\sup_{c \in C} f(c) \geq f(c^*)-\frac{1}{2}\epsilon$
	which is the desired contradiction.
\end{proof}	

\begin{theorem}
	In a market with a risk-free asset with risk-free-rate $r\geq 0$
	which satisfies condition \eqref{eq:finiteUtility} and for exponential preferences with $u$
	of the form \eqref{eq:crra} we have
	\[
	v_\infty = \lim_{n\to \infty} v_n
	\]
	where $v_n$ and $v_\infty$ are defined by equation \eqref{eq:fundObjective}.
\end{theorem}
\begin{proof}
	Let us write $v_\infty(B)$ for the value function for
	the problem with $n=\infty$ as a function
	of the budget $B$. For exponential preferences $v_\infty(B)$ is bounded above by $0$. We can always follow a constant-consumption-per-initial-individual strategy combined
	with investment in the risk-free account to get a strategy for which
	$v_\infty(B)$ is finite. As a function of $B$, $v_\infty(B)$ is concave.
	Thus $v_\infty$ is a continuous function on $(0,\infty)$.
	
	Hence given $\epsilon^*>0$ we can find $\tilde{\delta}>0$
	such that $v_\infty(B-\tilde{\delta})\geq v_\infty(B) - \tfrac{1}{2}\epsilon^*$.
	We may choose an
	investment-consumption strategy $(\tilde{\gamma}^1,\tilde{\alpha}^1)$
	for the problem with $n=\infty$ and budget $B-\delta$ such that
	\[
	{\cal J}(\tilde{\gamma}^1,\tau) \geq v_\infty(B-\tilde{\delta}) - \tfrac{1}{2} \epsilon^*.
	\]
	By adding on a constant-consumption-per-initial-individual strategy
	of cost less than $\tilde{\delta}$ we can find an investment-consumption
	strategy $(\gamma^1,\bm{\alpha}^1)$ for the problem with $n=\infty$ and
	budget $B$ such that
	\[
	{\cal J}(\gamma^1,\tau) \geq v_\infty - \epsilon^*.
	\]
	and moreover, the consumption per initial individual, and hence the consumption per survivor, never drops below $\delta:=C_1 \tilde{\delta}$ for an appropriate constant $C_1$.
	
	Let $\gamma_\epsilon$ be the strategy obtained from $\gamma^1_t$ by
	reducing consumption at time $T-\epsilon$ as described in Lemma \ref{lemma:gammaEpsilon}.
	
	Let $(\gamma^2,\bm{\alpha}^2)$ be
	the investment-consumption strategy for the investment problem for finite
	$n$ of consuming a fixed amount $c$ per initial individual at each time, and investing only in the risk-free asset. The consumption per survivor
	in this strategy will always be at least $c$. This strategy will cost
	a finite amount to implement. We assume that $c$ is chosen to ensure that
	the cost of strategy $(\gamma^2,\bm{\alpha}^2)$ is equal to the budget.

	Given parameters $\Lambda = (\lambda_1, \lambda_2, \epsilon) \in [0,1] \times [0,1] \times [0,T)$ we define an investment-consumption strategy $(\gamma^\Lambda, \bm{\alpha}^\Lambda)$ for the investment problem for finite $n$
	as follows. Divide the initial budget, $B$ into two accounts, account 1 and account 2,
	allocating a quantity $B \lambda_i$ to account $i$. From account 1,
	consume an amount $(\lambda_1)^2 (\gamma_\epsilon)_t$ per survivor at each time $t$ unless that account has run out of money. From account 2,
	consume an amount $\lambda_2 \gamma^2_t$ at each time $t \leq T$. In account $i$ invest in the proportions $\bm{\alpha}^i_t$ at each time $t$. Note
	that account 2 will never run out of money and hence,
	so long as $\lambda_1+\lambda_2 \leq 1$, this will be an admissible strategy.
	
	Let us write
	\[
	W^\Lambda = -\exp\left( - \int_0^{\tau_\iota} u(\gamma^\Lambda_t) \, \DT \right)
	\]
	so that the gain function is given by
	\[
	{\cal J}_\iota(\gamma^\Lambda, \tau_\iota) = \E_{\Omega_n}(W^\Lambda)
	\]
	The subscript ${\Omega_n}$ emphasizes that the expectation is taken
	in the probability space for the problem of $n$ individuals.
	
	Let ${\cal G}$ denote the event that the proportion of individuals
	surviving is less than or equal to $\frac{1}{\lambda_1} s_t$ at all times up to $T-\epsilon$. Because we only consume an amount $(\lambda_1)^2 (\gamma_\epsilon)_t$ per survivor from account 1, account 1 will not run
	out of money before time $T^*-\epsilon$ in the event ${\cal G}$, indeed
	by time $T^*-\epsilon$ we will still have sufficient funds to pay $\delta$
	per initial individual, and hence at least $\delta$ per survivor,
	until time $T^*$. Thus account 1 never runs out of funds in the event ${\cal G}$. We
	define the gain function conditioned on ${\cal G}$ by
	\[
	j(n,\Lambda):= \E_{\Omega_n} \left( W^\Lambda \mid \cal G \right).
	\]
	However, conditioned on ${\cal G}$, the consumption per survivor is independent of $n$ and in fact
	\[
	j(n,\Lambda) = j(\Lambda):= \E_{\Omega_\infty} \left( -\exp\left(- \int_0^{\tau_\iota} u(\lambda_1^2 (\gamma_\epsilon)_t + \lambda_2 \gamma^2_t) \, \DT \right) \right).
	\]
	Note that 
	\[
	j((1,0,0))={\cal J}(\gamma^1,\tau).
	\]
	
	Let us write $\mu_1=\lambda_1^2$ and $\mu_2=\lambda_2$.
	The set of values for $\mu_i$ which
	give admissible strategies is given by $N=\{(\mu_1,\mu_2):\sqrt{\mu_1} + \mu_2 \leq 1\}$. This contains the open convex set $C=\{ (\mu_1,\mu_2): \mu_1+2\mu_2 < 1\}$ which has the point $(1,0)$ on its boundary.
	
	For fixed $\epsilon$ define
	\[
	j_\epsilon(\mu_1,\mu_2) = j((\sqrt{\mu_1},\mu_2,\epsilon)).
	\]
	This is a concave function on $[0,1]\times[0,1]$. It follows by Lemma \eqref{eq:concavityLemma}
	that 
	\[
	\sup_{(\mu_1,\mu_2) \in C} j_\epsilon(\mu_1,\mu_2)
	\geq j(1,0,\epsilon).
	\]
	Hence if we define $C^\prime:=\{ (\lambda_1, \lambda_2): \lambda_1^2 + 2 \lambda_2 < 1 \}$, we will have
	\[
	\sup_{\Lambda \in C^\prime \times (0,T)} j(\Lambda) \geq 
	\sup_{\epsilon>0} j(1,0,\epsilon).
	\]
	Hence by Lemma \ref{lemma:gammaEpsilon} we have
	\[
	\sup_{\Lambda \in C^\prime \times (0,T)} j(\Lambda) \geq 
	j((1,0,0)) \geq v_\infty-\epsilon^*.
	\]
	
	It follows that we may find
	$\Lambda = (\lambda_1,\lambda_2,\epsilon) \in C^\prime$ such that
	\[
	j(\Lambda)\geq v_\infty-2 \epsilon^*.
	\]
	For the strategy $(\gamma^\Lambda, \alpha^\Lambda)$ we have
	\[
	\begin{split}
	{\cal J}(\gamma^\Lambda, \tau_\iota) &= 
	\E_{\Omega_n}(W \mid {\cal G})\, \P_{\Omega_n}( {\cal G})
	+ \E_{\Omega_n}(W \mid \lnot {\cal G})\, \P_{\Omega_n}( \lnot {\cal G}) \\
	&= j(\Lambda)\, \P_{\Omega_n}( {\cal G})
	+ \E_{\Omega_n}(W \mid \lnot {\cal G})\, \P_{\Omega_n}( \lnot {\cal G}).
	\end{split}
	\]
	Since $\lambda^2>0$, $\E_{\Omega_n}(W \mid \lnot {\cal G})>0$
	because the cashflow received is always bounded below by $\lambda_2 c$.
	
	By Lemma \ref{lemma:finiteTimePoints}, $\P_{\Omega_n}( {\cal G})$ tends to 1 as $n\to \infty$.
	Hence
	\[
	\lim_{n\to \infty} {\cal J}(\gamma^\Lambda, \tau_\iota)
	\geq v_\infty-2 \epsilon^*.
	\]
	So for sufficiently large $n$, $v_n \geq v_\infty-3 \epsilon^*$.
	This is true for arbitrary $\epsilon^*$,
	so $v_n \geq v_\infty$.
\end{proof}

This completes the proof that large collective funds are well approximated by an infinite fund.

\bibliography{collectivization}
\bibliographystyle{plain}

\end{document}